\documentclass[12pt,a4paper]{article}

\usepackage[cp 1251]{inputenc}
\usepackage{amssymb,amsmath,amsthm,euscript}
\usepackage{hyperref}
\usepackage[dvips]{graphicx}

\def\Hinf{\mbox{${\mathcal{H}_\infty}$}}
\def\H2{\mbox{${\mathcal{H}_2}$}}

\def\R{\mathbb{R}}

\def\MA{\overline{\bf A}}
\def\GG{\mathcal{G}}

\def\E{\mathbf{E}}
\def\e{\mathrm{e}}
\def\c{\mathrm{c}}

\def\sn{\mbox{$| \! | \! |$}}

\def\ol{\overline}
\def\wh{\widehat}

\def\bA{\bar{A}}
\def\bB{\bar{B}}
\def\bC{\bar{C}}
\def\bS{\bar{S}}
\def\bR{\bar{R}}
\def\bPhi{\bar{\Phi}}
\def\bPi{\bar{\Pi}}

\def\euA{\EuScript{A}}
\def\euB{\EuScript{B}}
\def\euC{\EuScript{C}}
\def\euD{\EuScript{D}}

\def\euAc{\EuScript{A}_\c}
\def\euBc{\EuScript{B}_\c}
\def\euCc{\EuScript{C}_\c}
\def\euDc{\EuScript{D}_\c}

\def\calA{\mathcal{A}}
\def\calB{\mathcal{B}}
\def\calC{\mathcal{C}}
\def\calD{\mathcal{D}}

\newcommand{\rank}{\mathop{\mathrm{rank}}\nolimits}

\newcommand{\esssup}{\mathop{\mathrm{ess\,sup}}}

\newcommand{\tr}{\mathop{\mathrm{tr}}\nolimits}

\newcommand{\blockdiag}{\mathop{\mathrm{blockdiag}}\nolimits}
\newcommand{\T}{\mathrm{T}}

\newtheorem{theorem}{Theorem}
\newtheorem{lemma}{Lemma}
\newtheorem{corollary}{Corollary}
\newtheorem{problem}{Problem}

\newtheorem{remark}{Remark}

\def\be#1{\begin{equation}\label{#1}}
\def\ee{\end{equation}}

\begin{document}
\title{Synthesis of Anisotropic Suboptimal Controllers\linebreak by Convex
Optimization\footnote{This work was supported
    by the Russian Foundation for Basic Research (grant 11-08-00714-a) and Program for Fundamental Research No.~15
of EEMCP Division of Russian Academy of Sciences.}}
\author{Michael M.~Tchaikovsky$^\dag$\quad Alexander P.
Kurdyukov$^\dag$\quad Victor N. Timin\footnote{The authors are
with Institute of Control Sciences of Russian Academy of Sciences,
Moscow, Russia, 117997, Profsoyuznaya 65, fax:
+7\,495\,334\,93\,40, e-mails: mmtchaikovsky@hotmail.com,
akurd@ipu.ru,\newline timin.victor@rambler.ru.}}
\date{\today}
\maketitle

\begin{abstract}
This paper considers a disturbance attenuation problem for a
linear discrete time invariant system under random disturbances
with imprecisely known distributions. The statistical uncertainty
is measured in terms of the mean anisotropy functional. The
disturbance attenuation capabilities of the system are quantified
by the anisotropic norm which is a stochastic counterpart of the
$\Hinf$ norm. The designed anisotropic suboptimal  controller
generally is a dynamic fixed-order output-feedback compensator
which is required to stabilize the closed-loop system and keep its
anisotropic norm below a prescribed threshold value.  Rather than
resulting in a unique controller, the suboptimal design procedure
yields a family of controllers, thus providing freedom to impose
some additional performance specifications on the closed-loop
system. The general fixed-order synthesis procedure employs
solving a convex inequality on the determinant of a positive
definite matrix and two linear matrix inequalities in reciprocal
matrices which make the general optimization problem nonconvex. By
applying the known standard convexification procedures it is shown
that the resulting optimization problem is convex for the
full-information state-feedback, output-feedback full-order
controllers, and static output-feedback controller for some
specific classes of plants defined by certain structural
properties. In the convex cases, the anisotropic $\gamma$-optimal
controllers are obtained by minimizing the squared norm threshold
value subject to convex constraints. In a sense, the anisotropic
controller seems to offer a promising and flexible trade-off
between $\H2$ and $\Hinf$ controllers which are its limiting
cases. In comparison with the state-space solution to the
anisotropic optimal controller synthesis problem presented before
which results in a unique full-order estimator-based controller
defined by a complex system of cross-coupled nonlinear matrix
algebraic equations, the proposed optimization-based approach is
novel and does not require developing specific homotopy-like
computational algorithms.
\end{abstract}


\section{Introduction}

The stochastic uncertainty of random disturbances regarded as a
discrepancy between an inexactly known probability distribution of
a real-world noise and its nominal model can significantly degrade
the designed performance of a control system if the applied
controller synthesis procedure relies upon a specific probability
law of the disturbance and the assumption that it is known
precisely. Such situations can also result from the inherent
variability of the conditions of the control system operational
environment. So, the $\H2$ and $\Hinf$ controllers are efficient
in full only if the basic hypotheses on the nature of external
disturbances are met closely enough. It is known that the $\H2$
(or LQG) controller may perform poorly if the input disturbance is
a strongly correlated noise~\cite{Doyle_1978}, while the $\Hinf$
controller designed for the deterministic worst
case~\cite{DGKF_1989} demonstrates excessive conservatism if the
external disturbance is white or weakly correlated random signal.

One of the first ideas aimed at overcoming the lack of performance
of the LQG controller in the case when the external disturbance is
not the Gaussian white noise arose in
work~\cite{Jacobson_1977_book} devoted to some modification of the
performance criterion. This idea gave rise to development of the
whole class of problems in the control theory called the risk
sensitivity problems~\cite{Whittle_1981,Whittle_1989}.

The ideas of deriving controller which combines the positive
features of LQG ($\H2$) and $\Hinf$ controllers (i.e. minimizes
the quadratic cost sufficiently good and is robust enough)
appeared in the beginning of 1990's. In particular, one can
distinguish an approach concerned with minimization of $\H2$ norm
of the closed-loop system under constraints on its $\Hinf$
norm~\cite{BH_1989} and approach related to minimization of
$\Hinf$ entropy functional under constraints on the closed-loop
$\Hinf$ norm~\cite{MG_1991}.

As is shown in~\cite{Glover&Doyle_1988}, the problem of synthesis
of a controller which minimizes the $\Hinf$ entropy functional is
equivalent in a sense to the problem of optimal risk-sensitive
(LEQG) controller synthesis. A lot of papers are devoted to the
problems concerned with minimization of the $\Hinf$ entropy
functional (see
e.g.~\cite{Iglesias&Mustafa&Glover_1990}--\cite{Fridman&Shaked_2000}).

The ideas of the mixed $\H2/\Hinf$ control first introduced
in~\cite{BH_1989} were extended in~\cite{ZGBD_1994a,RS_1998} based
on splitting the external disturbance into signals with bounded
spectrum and bounded power and using the multi-objective
$\H2/\Hinf$ performance criterion. A solution to the stochastic
mixed $\H2/\Hinf$ control problem for the discrete-time systems is
given in~\cite{Miradore&Ricci_2005}.

All of the works mentioned above exploit the techniques based on
solving certain (sometimes cross-coupled) Riccati equations.
In~\cite{Khargonekar&Rotea_1991} the mixed $\H2/\Hinf$ problem was
considered in terms of algebraic Riccati inequalities rather than
equations and solved by means of convex optimization. Since then
the efficient interior-point algorithms for solving convex
optimization problems had been
developed~\cite{BG_1993}--\cite{NG_1994}, 
convex optimization has become a standard strategy for control
system analysis and synthesis. The linear matrix inequalities have
proved to be a powerful formulation and design technique for a
variety of linear problems~\cite{Boyd_book_1994}. After the
$\Hinf$ controller synthesis problem had been solved via
LMI~\cite{IS_1994,GA_1994}, the semidefinite programming was
successfully applied to developing effective solutions to
multi-objective $\H2/\Hinf$ control
problems~\cite{Scherer_1995}--\cite{Yu_2004}. 
A detailed survey of these extensive results is far beyond the
topic of this paper and may be presented elsewhere.

An approach to attenuation of uncertain stochastic disturbances
based on minimax control was proposed in the middle of
1990's~\cite{Gusev_1995a}--\cite{Gusev_1996} 
and  extended later to the MIMO systems and synthesis of
structured controllers via LMI in~\cite{Scherer_2000_stoch}.
Instead of exact knowledge of the disturbance's covariance
coefficients, it is only required that the covariance coefficients
belong to an a priori known set. The designed controller minimizes
the worst possible asymptotic output variance for all these
disturbances. The considered problem is intermediate between the
extreme $\H2$ and $\Hinf$ design scenarios and reduces to a robust
control problem with uncertainty in the external disturbance
signal~\cite{Scherer_2000_stoch}.

At the same time, another promising stochastic minimax alternative
had emerged from ideas of I.G.\,Vladimirov who originally
developed the anisotropy-based theory of robust stochastic
control presented in a series of
papers~\cite{SVK_1994}--\cite{VKS_1996_1}. 
In the view of this approach, the robustness in stochastic control
is achieved by explicitly incorporating different scenarios of the
noise distribution into a single performance index to be
optimized; the statistical uncertainty is measured in entropy
theoretic terms, and the robust performance index can be chosen so
as to quantify the worst-case disturbance attenuation capabilities
of the system. The main concepts of the anisotropy-based approach
to robust stochastic control are the anisotropy of a random vector
and anisotropic norm of a system.

The anisotropy functional introduced by I.G.\,Vladimirov is an
entropy theoretic measure of the deviation of a probability
distribution in Euclidean space from Gaussian distributions with
zero mean and scalar covariance matrices. The mean anisotropy of a
stationary random sequence is defined as the anisotropy production
rate per time step for long segments of the sequence. In
application to random disturbances, the mean anisotropy describes
the amount of statistical uncertainty which is understood as the
discrepancy between the imprecisely known actual noise
distribution and the family of nominal models which consider the
disturbance to be a stationary Gaussian white noise sequence with
a scalar covariance matrix~\cite{VKS_1996_1,DVKS_2001}.

Another fundamental concept of I.G.\,Vladimirov's theory is the
$a$-anisotropic norm of a linear discrete time invariant (LDTI)
system which quantifies the disturbance attenuation capabilities
by the largest ratio of the power norm of the system output to
that of the input provided that the mean anisotropy of the input
disturbance does not exceed a given nonnegative level
$a$~\cite{VKS_1996_1,DVKS_2001}. A generalization of the
anisotropy-based robust performance analysis to finite horizon
time varying systems is developed in~\cite{VDK_2006}.

In the context of robust stochastic control design aimed at
suppressing the potentially harmful effects of statistical
uncertainty, the anisotropy-based approach offers an important
alternative to those control design procedures that rely on a
precisely known specific probability law of the disturbance and
the assumption that it is known precisely. Minimization of the
anisotropic norm of the closed-loop system as a performance
criterion leads to internally stabilizing dynamic output-feedback
controllers that are less conservative than the $\Hinf$
controllers and more efficient for attenuating the correlated
disturbances than the $\H2$ controllers~\cite{DVKS_2001}. A
state-space solution to the anisotropic optimal control problem
derived by I.G.\,Vladimirov in~\cite{VKS_1996_2} involves the
solution of three cross-coupled algebraic Riccati equations, an
algebraic Lyapunov equation and an equation on the determinant  of
a related matrix. The resulted optimal full-order estimator-based
(central) controller is a unique one. An extension of these
results to the systems with parametric uncertainties was given
in~\cite{KM_2005_IFAC,KM_2006_AiT}. But solving these complex
systems of equations  requires special developing of homotopy-like
numerical algorithms~\cite{DKSV_1997_report}. Besides, the applied
equation-based synthesis procedure is not aimed at the synthesis
of reduced- or fixed-order (decentralized, structured,
multi-objective) controllers which still remains open. Moreover,
although the ideas of entropy-constrained induced norms and
associated stochastic minimax find further development in the
control literature~\cite{CR_2007}, the anisotropy-based theory of
stochastic robust control remains largely unnoticed. One of the
reasons seems to be hard numerical tractability of the analysis
and synthesis problems as well as a lack of additional degrees of
freedom in the controller synthesis procedure.

The  anisotropic suboptimal controller design is a natural
extension of the approach proposed by I.G.\,Vladimirov
in~\cite{VKS_1996_2}. Instead of minimizing the anisotropic norm
of the closed-loop system, a suboptimal controller is only
required to keep it below a given threshold value. Rather than
resulting in a unique controller, the suboptimal synthesis yields
a family of controllers, thus providing freedom to impose some
additional specifications on the closed-loop system. One of such
specifications, for example, may be a particular pole placement to
achieve desirable transient performance. Getting a solution to the
anisotropic suboptimal controller synthesis problem requires a
state-space criterion to verify whether the anisotropic norm of a
system does not exceed a given value. An Anisotropic Norm Bounded
Real Lemma (ANBRL) as a stochastic counterpart of the well-known
$\Hinf$ norm Bounded Real Lemma for LDTI systems under
statistically uncertain stationary Gaussian random disturbances
with limited mean anisotropy was presented in~\cite{KMT_2010}. The
resulting criterion has the form of an inequality on the
determinant of a matrix associated with an algebraic Riccati
equation which depends on a scalar parameter. A similar criterion
for linear discrete time varying systems involving a
time-dependent inequality and difference Riccati equation is
derived in~\cite{MKV_2011}. Recently, a sufficient strict version
of ANBRL was introduced in~\cite{TKT_2011_IFAC,TK_2011_prep} in
form of a convex feasibility problem employing a strict inequality
in the determinant of a positive-definite matrix and a related
LMI. Moreover, the determinant constraint turns out to depend
linearly on the squared threshold value, thus allowing to minimize
it directly subject to the convex constraints and compute the
$a$-anisotropic norm of a LDTI system as a solution to the convex
optimization problem~\cite{TK_2011_prep}. The developed analysis
procedure is numerically attractive and easily realizable by means
of available convex optimization
software~\cite{Sturm_1999,Lofberg_2004}. This paper is aimed at
application of the powerful technique of convex optimization to
synthesis of the anisotropic suboptimal and $\gamma$-optimal
controllers generally of fixed order. The anisotropic controller
seems to offer a promising and flexible trade-off between $\H2$
and $\Hinf$ controllers. In comparison with the state-space
solution to anisotropic optimal controller synthesis problem
derived before in~\cite{VKS_1996_2}, the proposed
optimization-based approach is novel and does not require
developing specific homotopy-like computational
algorithms~\cite{DKSV_1997_report}.

The structure of the paper is as follows. In
Section~\ref{sect:problem statement} we give the statement of the
general problem of synthesis of the fixed-order anisotropic
suboptimal controller. In Section~\ref{sect:problem solution} we
introduce a solution to the general fixed-order synthesis problem
and consider three important design cases: static state-feedback
gain for full-information case, dynamic output-feedback
controller, and static output-feedback
 gain. Section~\ref{sect:numerical example} provides a
number of illustrative numerical examples. Concluding remarks are
given in Section~\ref{sect:conclusion}.

\subsection{Notation}

The set of reals is denoted by $\R,$ the set of real $(n\times
m)$- matrices is denoted by $\R^{n\times m}.$ For a complex matrix
$M = [m_{ij}]$, $M^\ast$ denotes the Hermitian conjugate of the
matrix: $M^\ast:= [m^\ast_{ji}].$ For a real matrix $M =
[m_{ij}]$, $M^\T$ denotes the transpose of the matrix: $M^\T :=
[m_{ji}].$ For real symmetric matrices, $M\succ N$ stands for
positive definiteness of $M-N.$ In block symmetric matrices,
symbol $\ast$ replaces blocks that are readily inferred by
symmetry. The spectral radius of a matrix $M$ is denoted by
$\rho(M):=\max_k|\lambda_k(M)|,$ where $\lambda_k(M)$ is $k$-th
eigenvalue of the matrix $M.$ The maximum singular value of a
complex matrix $M$ is denoted by
$\ol{\sigma}(M):=\sqrt{\lambda_{\max}(M^\ast M)}.$ $I_n$ denotes
$(n\times n)$ identity matrix, $0_{n\times m}$ denotes zero
$(n\times m)$ matrix. The dimensions of zero matrices, where they
can be understood from the context, will be omitted for the sake
of brevity.

The angular boundary value of a transfer function $F(z)$ analytic
in the unit disc of the complex plane $|z|<1$ is denoted by
$$
{\widehat F}(\omega) :=\lim_{r\rightarrow 1-}{F(r\e^{i\omega})}.
$$
$\mathcal{H}_2^{p\times m}$ denotes the Hardy space of ($p\times
m$)-matrix-valued transfer functions $F(z)$ of a complex variable
$z$ which are analytic in the unit disc $|z|<1$ and have bounded
$\H2$ norm
$$
\|F\|_2 := \left(
\frac{1}{2\pi}\int\limits_{-\pi}^{\pi}{\tr(\wh{F}(\omega)\wh{F}^\ast(\omega))d\omega}
\right)^{1/2}.
$$
$\mathcal{H}_\infty^{p\times m}$ denotes the Hardy space of
($p\times m$)-matrix-valued transfer functions $F(z)$ of a complex
variable $z$ which are analytic in the unit disc $|z|<1$ and have
bounded $\Hinf$ norm
$$
\|F\|_\infty := \sup_{|z|\geqslant1}\ol{\sigma}(F(z)) =
\esssup_{-\pi\leqslant\omega\leqslant\pi}\ol{\sigma}(\wh{F}(\omega)).
$$

\section{Problem statement}\label{sect:problem statement}

Consider a LDTI plant $P(z)$ with $n_x$-di\-men\-si\-o\-nal
internal state $X,$ $m_w$-di\-men\-si\-o\-nal disturbance input
$W,$ $m_u$-di\-men\-si\-o\-nal control input $U,$
$p_z$-di\-men\-si\-o\-nal controlled output $Z,$ and
$p_y$-di\-men\-si\-o\-nal measured output $Y.$ All these signals
are double-sided discrete-time sequences related to each other by
the equations
\begin{equation}\label{eq:standard plant}
P(z):\,\,\left[
\begin{array}{c}
x_{k+1}\\ z_k\\ y_k
\end{array}
\right] = \left[
\begin{array}{ccc}
A & B_w & B_u\\ C_z & D_{zw} & D_{zu}\\ C_y & D_{yw} & 0
\end{array}
\right] \left[
\begin{array}{c}
x_{k}\\ w_k\\ u_k
\end{array}
\right],\quad
     -\infty < k < +\infty,
\end{equation}
where all matrices are assumed to be of appropriate dimensions and
$p_z \leqslant m_w;$  $(A,B_u)$ and $(A,C_y)$ are assumed to be
stabilizable and detectable.

The only prior information on the probability distribution of the
disturbance sequence $W = (w_k)_{-\infty<k<+\infty}$ is as
follows. It is assumed that $W$ is a stationary sequence of random
vectors $w_k$ with zero mean $\E w_k = 0,$ unknown covariance
matrix $\E w_kw_k^\T  = \Sigma_W\succ0,$ and Gaussian PDF
 $$
p(w_k) :=
(2\pi)^{-m_1/2}(\det{\Sigma_W})^{-1/2}\exp\left(-\frac{1}{2}\|w_k\|_{\Sigma_W^{-1}}^2\right),
 $$
where $\|w_k\|_{\Sigma_W^{-1}} = \sqrt{w_k^\T\Sigma_W^{-1}w_k}$
and $\E$ denotes the expectation. It is also assumed that the mean
anisotropy of the sequence $W$ is bounded by a nonnegative
parameter $a$. The latter means that $W$ can be produced from
$m_w$-dimensional Gaussian white noise $V=
(v_k)_{-\infty<k<+\infty}$ with zero mean $\E v_k=0$ and scalar
covariance matrix $\E v_kv_k^\T = \lambda
  I_{m_1},$ $\lambda>0$, by an unknown stable LTI shaping filter $G(z)$
  in the family
\begin{equation*}\label{eq:G_alpha def}
\GG_a := \left\{ G \in \mathcal{H}_2^{m\times m}\colon\ \MA(G)
\leqslant a \right\},
\end{equation*}
where
\begin{equation*}\label{eq:meananiso functional}
    \MA(G) = -\frac{1}{4\pi}\, \int\limits_{-\pi}^{\pi} \ln\det
    \left(\frac{m_w}{\|G\|_2^2}\, {\widehat G}(\omega) {\widehat
    G}^\ast(\omega)\right)\, d\omega
\end{equation*}
is the mean anisotropy functional~\cite{VKS_1996_1,DVKS_2001}.

We are generally interested in finding a fixed-order dynamic
output-feedback controller in general compensator form
\begin{equation}\label{eq:dynamic controller}
K(z):\,\,\left[
\begin{array}{c}
\xi_{k+1}\\ u_k
\end{array}
\right] = \left[
\begin{array}{cc}
A_\c & B_\c\\ C_\c & D_\c
\end{array}
\right]\left[
\begin{array}{c}
\xi_k\\ y_k
\end{array}
\right],\quad -\infty < k < +\infty,
\end{equation}
with $n_\xi$-dimensional internal state
$\Xi=(\xi_k)_{-\infty<k<+\infty}$ to ensure stability of the
closed-loop system (Figure~\ref{fig:closed-loop system}) and
guarantee some designed level of the external disturbance
attenuation performance.

\begin{figure}[thpb]
      \centering
      \includegraphics[scale=1]{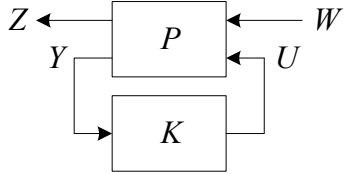}
      \caption{Closed-loop system}\label{fig:closed-loop system}
\end{figure}

Let $T_{ZW}(z)$ denote the closed-loop transfer function from $W$
to $Z$. Recall that the $a$-anisotropic norm of a transfer
function $T_{ZW}(z)\in\mathcal{H}_\infty^{p_z\times m_w}$
quantifies the disturbance attenuation capabilities of the
respective closed-loop system by the largest ratio of the power
norm of the system output to that of the input provided that the
mean anisotropy of the input disturbance does not exceed the level
$a$~\cite{VKS_1996_1,DVKS_2001}:
\begin{equation}\label{eq:aninorm def}
\sn T_{ZW}\sn_a :=
\sup_{G\in\GG_a}{\frac{\|T_{ZW}G\|_2}{\|G\|_2}}.
\end{equation}
Moreover, it is known from~\cite{VKS_1996_1,DVKS_2001} that the
$a$-anisotropic norm of a given system
$F\in\mathcal{H}_\infty^{p_z\times m_w}$ is a nondecreasing
continuous function of the mean anisotropy level $a$ which
satisfies
\begin{equation}
\label{limits}
    \frac{1}{\sqrt{m_w}}
    \|T_{ZW}\|_2
    =
    \sn T_{ZW} \sn_0
    \leqslant
    \lim_{a\to+\infty}
    \sn T_{ZW} \sn_a
    =
    \|T_{ZW}\|_\infty.
\end{equation}
These relations show that the $\H2$ and $\Hinf$ norms are the
limiting cases of the $a$-anisotropic norm  as $a\to 0, +\infty$,
respectively.

The statement of the general problem of synthesis of the
fixed-order anisotropic suboptimal controller is as follows.

\begin{problem}\label{problem:anisotropic suboptimal design}
Given a LDTI plant $P$ described by~(\ref{eq:standard plant}), a
mean anisotropy level $a\geqslant 0$ of the external disturbance
$W$, and some designed threshold value $\gamma>0$, find a
fixed-order LDTI output-feedback controller $K$ defined
by~(\ref{eq:dynamic controller}) which internally stabilizes the
closed-loop system and ensures its $a$-anisotropic norm does not
exceed the threshold $\gamma$, i.e.
\begin{equation}\label{eq:anisotropic suboptimality condition}
\sn T_{ZW}\sn_a < \gamma.
\end{equation}
\end{problem}

\section{Problem solution}\label{sect:problem solution}

Here we introduce a solution to the general fixed-order synthesis
problem and consider three important design cases, namely static
state-feedback gain for full-information case, dynamic
output-feedback controller, and static output-feedback gain. To
solve the synthesis problem, we apply a state-space criterion to
verify if the anisotropic norm  of a system does not exceed a
given threshold value. This criterion called the Strict
Anisotropic Norm Bounded Real Lemma (SANBRL) was recently
presented in~\cite{TKT_2011_IFAC,TK_2011_prep}. But to apply
SANBRL to the synthesis problem we should recast it in slightly
different form.

\subsection{Anisotropic norm bounded real lemma}\label{subsect:BRL}

With the plant $P$ and controller $K$ defined as above, the
closed-loop system admits the realization
\begin{equation}\label{eq:CL system equations}
T_{ZW}(z):\,\,\left[
\begin{array}{c}
\chi_{k+1}\\ z_k
\end{array}
\right] = \left[
\begin{array}{cc}
\EuScript{A} & \EuScript{B}\\ \EuScript{C} & \EuScript{D}
\end{array}
\right]\left[
\begin{array}{c}
\chi_k\\ w_k
\end{array}
\right],\quad -\infty < k < +\infty,
\end{equation}
where $\chi_k\in\R^n$, $n = n_x+n_\xi$.

It is shown in~\cite{TKT_2011_IFAC,TK_2011_prep} that given
$a\geqslant 0$, $\gamma > 0$, the inequality~(\ref{eq:anisotropic
suboptimality condition}) holds true if there exists
$\eta>\gamma^2$ such that the inequality
\begin{equation}
    \label{eq:determinant inequality gamma2 linear}
    \eta-(\det{(\e^{-2a/m_w}(\eta I_{m_w}-\EuScript{B}^\T \Phi\EuScript{B}-\EuScript{D}^\T \EuScript{D}))})^{1/m_w}
    < \gamma^2
\end{equation}
holds for a real $(n\times n)$-matrix $\Phi=\Phi^\T\succ 0$
satisfying LMI
\begin{equation}\label{eq:aninorm LMI 2x2 in Phi eta}
\left[
\begin{array}{cc}
\EuScript{A}^\T\Phi\EuScript{A}-\Phi+\EuScript{C}^\T\EuScript{C} & \EuScript{A}^\T\Phi\EuScript{B}+\EuScript{C}^\T\EuScript{D}\\
\EuScript{B}^\T\Phi\EuScript{A}+\EuScript{D}^\T\EuScript{C} &
\EuScript{B}^\T\Phi\EuScript{B}+\EuScript{D}^\T\EuScript{D}-\eta
I_{m_w}
\end{array}
\right] \prec 0.
\end{equation}
Note that the constraints described by the
inequalities~(\ref{eq:determinant inequality gamma2 linear}) and
(\ref{eq:aninorm LMI 2x2 in Phi eta}) are convex with respect to
both variables $\eta$ and $\Phi.$ Indeed, the function
$-(\det(\cdot))^{1/m_w}$ of a positive definite $(m_w\times
m_w)$-matrix on the left-hand side of~(\ref{eq:determinant
inequality gamma2 linear}) is convex; see~\cite{NN_1994,BTN_2000}.

Being convex in both variables $\eta$ and $\Phi$, the
conditions~(\ref{eq:determinant inequality gamma2 linear}),
(\ref{eq:aninorm LMI 2x2 in Phi eta}) of
SANBRL~\cite{TKT_2011_IFAC,TK_2011_prep} are not directly
applicable to solving the intended synthesis problem because of
the cross-products of the unknown Lyapunov matrix $\Phi$ and the
closed-loop realization matrices
$(\EuScript{A},\EuScript{B},\EuScript{C},\EuScript{D})$ depending
affinely on the controller parameters, which also appear
in~(\ref{eq:determinant inequality gamma2 linear}). Moreover, just
the inequality~(\ref{eq:determinant inequality gamma2 linear})
does not allow for the well-known Projection
Lemma~\cite{GA_1994,IS_1994} to be applied to get rid of the
controller realization matrices in the synthesis inequalities.

To overcome this obstacle, let us first move the positive definite
matrix $\eta I_{m_w}-\EuScript{B}^\T
\Phi\EuScript{B}-\EuScript{D}^\T \EuScript{D}$ away from the
determinant in~(\ref{eq:determinant inequality gamma2 linear}) by
introducing a slack variable, real $(m_w\times m_w)$-matrix
$\Psi=\Psi^\T\succ 0$ such that
\begin{equation}\label{eq:one-constraint S-procedure inequalities}
    \eta-(\det{(\e^{-2a/m_w}{\Psi})})^{1/m_w} < \gamma^2,\quad
    \Psi\prec\eta I_{m_w}-\euB^\T\Phi \euB-\euD^\T
    \euD
\end{equation}
which is equivalent to~(\ref{eq:determinant inequality gamma2
linear}). Then, let us decouple the cross-products of $\Phi$,
$\euB$, and $\euD$ in~(\ref{eq:one-constraint S-procedure
inequalities}). For this purpose, the latter inequality
in~(\ref{eq:one-constraint S-procedure inequalities}) can be
rewritten as
$$
\Psi-\eta I_{m_w}-\left[
\begin{array}{cc}
\euB^\T & \euD^\T
\end{array}
\right]\left[
\begin{array}{cc}
-\Phi^{-1} & 0\\
0 & -I_{p_z}
\end{array}
\right]^{-1}\left[
\begin{array}{c}
\euB\\ \euD
\end{array}
\right] \prec 0,
$$
where $\left[
\begin{array}{cc}
-\Phi^{-1} & 0\\
0 & -I_{p_z}
\end{array}
\right] \prec 0$, which is equivalent to
\begin{equation}\label{eq:Phi Psi eta LMI 3x3}
\left[
\begin{array}{ccc}
\Psi-\eta I_{m_w} & \euB^\T & \euD^\T\\
\euB & -\Phi^{-1} & 0\\
\euD & 0 & -I_{p_z}
\end{array}
\right]\prec 0
\end{equation}
by virtue of the Schur Theorem; see e.g.~\cite{Boyd_book_1994}.

To decouple the cross-products of $\Phi$, $\euA$, and $\euB$
in~(\ref{eq:aninorm LMI 2x2 in Phi eta}), represent it as
$$
\left[
\begin{array}{cc}
-\Phi+\euC^\T\euC & \euC^\T\euD\\
\euD^\T\euC & \-\eta I_{m_w}+\euD^\T\euD
\end{array}
\right]-\left[ \begin{array}{c} \euA^\T\\\euB^\T
\end{array} \right](-\Phi^{-1})^{-1}\left[ \begin{array}{cc} \euA &
\euB
\end{array} \right] \prec 0
$$
where $-\Phi^{-1}\prec 0$ evidently. Then by the Schur Theorem the
last inequality is equivalent to
\begin{equation}\label{eq:aninorm LMI 3x3 in Phi eta}
\left[
\begin{array}{ccc}
-\Phi+\euC^\T\euC & \euC^\T\euD &
\euA^\T\\
\euD^\T\euC & \euD^\T\euD-\eta I_{m_w} & \euB^\T\\
\euA & \euB & -\Phi^{-1}
\end{array}
\right]\prec 0.
\end{equation}
To decouple the cross-products of $\euC$ and $\euD$, let us
represent the inequality~(\ref{eq:aninorm LMI 3x3 in Phi eta}) as
$$
\left[
\begin{array}{ccc}
-\Phi & 0 &
\euA^\T\\
0 & -\eta I_{m_w} & \euB^\T\\
\euA & \euB & -\Phi^{-1}
\end{array}
\right]\\
-\left[
\begin{array}{c}
\euC^\T\\ \euD^\T\\ 0
\end{array}
\right](-I_{p_z})^{-1}\left[
\begin{array}{ccc}
\euC & \euD & 0
\end{array}
\right]\prec 0
$$
where $-I_{p_z}\prec 0$ clearly. Second application of the Schur
Theorem to the above inequality gives the following formulation of
SANBRL in reciprocal matrices.

\begin{lemma}\label{lemma:SANBRL inverse matrices}
Let $T_{ZW} \in \mathcal{H}_{\infty}^{p_z \times m_w}$ be a system
with the state-space realization~(\ref{eq:CL system equations}),
where $\rho(\EuScript{A})<1$. Then its $a$-anisotropic norm
(\ref{eq:aninorm def}) is strictly bounded by a given threshold
$\gamma>0$, i.e. $\sn T_{ZW}\sn_a<\gamma$ if there exists
$\eta>\gamma^2$ such that the inequality
\begin{equation}\label{eq:det inequality ANBRL inverse matrices}
\eta-(\det{(\e^{-2a/m_w}{\Psi})})^{1/m_w} < \gamma^2
\end{equation}
holds true for some real $(m_w\times m_w)$-matrix
$\Psi=\Psi^\T\succ 0$
 and $(n\times n)$-matrix
$\Phi=\Phi^\T\succ 0$ satisfying inequalities
\begin{equation}\label{eq:LMI 3x3 ANBRL inverse matrices}
\left[
\begin{array}{ccc}
\Psi-\eta I_{m_w} & \euB^\T & \euD^\T\\
\euB & -\Phi^{-1} & 0\\
\euD & 0 & -I_{p_z}
\end{array}
\right]\prec 0,
\end{equation}
\begin{equation}\label{eq:LMI 4x4 ANBRL inverse matrices}
\left[
\begin{array}{cccc}
-\Phi & 0 & \euA^\T & \euC^\T\\
0 & -\eta I_{m_w} & \euB^\T & \euD^\T\\
\euA & \euB & -\Phi^{-1} & 0\\
\euC & \euD & 0 & -I_{p_z}
\end{array} \right] \prec 0.
\end{equation}
\end{lemma}
Thus, with the notation $\Pi := \Phi^{-1}$, verifying if the
condition $\sn T_{ZW} \sn_a < \gamma$ holds true reduces to
finding a positive scalar $\eta$ and two matrices $\Phi\succ 0$,
$\Pi\succ 0$, $\Phi\Pi = I_n$, satisfying the LMIs~(\ref{eq:LMI
3x3 ANBRL inverse matrices}), (\ref{eq:LMI 4x4 ANBRL inverse
matrices}) under the convex constraint~(\ref{eq:det inequality
ANBRL inverse matrices}) or making sure of insolvability of this
problem. For solving this nonconvex problem numerically, one can
make use of known algorithms developed
in~\cite{IS_1995_XY_centering}--\cite{Polyak_Gryazina_2010}
suitable for finding reciprocal matrices under convex constraints.

\subsection{State-feedback controller}\label{subsect:state
feedback}

Before to proceed to general synthesis
Problem~\ref{problem:anisotropic suboptimal design}, let us
consider the full-information case, when the state vector can be
measured precisely and the plant is described by the equations
\begin{equation}\label{eq:plant full information}
P(z):\,\,\left[
\begin{array}{c}
x_{k+1}\\ z_k\\ y_k
\end{array}
\right] = \left[
\begin{array}{ccc}
A & B_w & B_u\\ C_z & D_{zw} & D_{zu}\\ I_{n_x} & 0 & 0
\end{array}
\right] \left[
\begin{array}{c}
x_{k}\\ w_k\\ u_k
\end{array}
\right],\quad
     -\infty < k < +\infty,
\end{equation}
where as above all matrices are assumed to be of appropriate
dimensions and $p_z \leqslant m_w;$ $(A,B_u)$ is assumed to be
stabilizable.

\begin{problem}\label{problem:aniso sub sf}
Given a LDTI plant $P$ described by~(\ref{eq:plant full
information}), a mean anisotropy level $a\geqslant 0$ of the
external disturbance $W$, and some designed threshold value
$\gamma>0$, find a static state-feedback  controller
\begin{equation}\label{eq:static sf controller}
u_k = Kx_k
\end{equation}
which internally stabilizes the closed-loop system $T_{ZW}(z)$
with the state-space realization
\begin{equation}\label{eq:CL realization sf}
\left[
\begin{array}{c|c}
\EuScript{A} & \EuScript{B}\\\hline \EuScript{C} & \EuScript{D}
\end{array}
\right] = \left[
\begin{array}{c|c}
A+B_uK & B_w\\\hline C_z+D_{zu}K & D_{zw}
\end{array}
\right]
\end{equation}
 and ensures its
$a$-anisotropic norm does not exceed the threshold $\gamma$, i.e.
the inequality~(\ref{eq:anisotropic suboptimality condition})
holds.
\end{problem}

The following theorem gives sufficient conditions for the static
state-feedback anisotropic suboptimal controller to exist.

\begin{theorem}\label{theorem:state-feedback problem solution}
Given $a\geqslant 0$, $\gamma>0$, the state-feedback
controller~(\ref{eq:static sf controller}) stabilizing the
closed-loop system~(\ref{eq:CL realization sf}) ($\rho(A+B_uK)<1$)
and ensuring~(\ref{eq:anisotropic suboptimality condition}) exists
if the convex problem
\begin{equation}\label{eq:det Psi inequality sf}
\eta-(\det{(\e^{-2a/m_w}{\Psi})})^{1/m_w} < \gamma^2,
\end{equation}
\begin{equation}\label{eq:Psi LMI 3x3 sf}
\left[
\begin{array}{ccc}
\Psi-\eta I_{m_w} & B_w^\T & D_{zw}^\T\\
B_w & -\Pi & 0\\
D_{zw} & 0 & -I_{p_z}
\end{array}
\right]\prec 0,
\end{equation}
\begin{equation}\label{eq:aninorm LMI 4x4 sf}
\left[
\begin{array}{cccc}
-\Pi & 0 & \Pi A^\T+\Lambda^\T B_u^\T & \Pi C_z^\T+\Lambda^\T D_{zu}^\T\\
0 & -\eta I_{m_w} & B_w^\T & D_{zw}^\T\\
A\Pi+B_u\Lambda & B_w & -\Pi & 0\\
C_z\Pi+D_{zu}\Lambda & D_{zw} & 0 & -I_{p_z}
\end{array} \right] \prec 0,
\end{equation}
\begin{equation}\label{eq:pos def vars sf}
\eta>\gamma^2,\quad \Psi\succ0,\quad \Pi\succ0
\end{equation}
is feasible with respect to the scalar variable $\eta$, real
$(m_w\times m_w)$-matrix $\Psi$, real $(n_x\times n_x)$-matrix
$\Pi$, and real $(m_u\times n_x)$-matrix~$\Lambda$. If the
problem~(\ref{eq:det Psi inequality sf})--(\ref{eq:pos def vars
sf}) is feasible and the unknown variables have been found, then
the state-feedback controller gain matrix is determined by $K =
\Lambda\Pi^{-1}$.
\end{theorem}
\begin{proof}
Let a solution to the problem~(\ref{eq:det Psi inequality
sf})--(\ref{eq:pos def vars sf}) exist. Define $\Phi := \Pi^{-1}$.
By definition of $K = \Lambda\Pi^{-1}$, the LMIs~(\ref{eq:Psi LMI
3x3 sf}), (\ref{eq:aninorm LMI 4x4 sf}) can be rewritten as
\begin{equation}\label{eq:LMI 3x3 sf}
\left[
\begin{array}{ccc}
\Psi-\eta I_{m_w} & B_w^\T & D_{zw}^\T\\
B_w & -\Phi^{-1} & 0\\
D_{zw} & 0 & -I_{p_z}
\end{array}
\right]\prec 0,
\end{equation}
$$
\left[
\begin{array}{cccc}
-\Phi^{-1} & 0 & \Phi^{-1}A^\T+\Phi^{-1}K^\T B_u^\T & \Phi^{-1}C_z^\T+\Phi^{-1}K^\T D_{zu}^\T\\
0 & -\eta I_{m_w} & B_w^\T & D_{zw}^\T\\
A\Phi^{-1}+B_uK\Phi^{-1} & B_w & -\Phi^{-1} & 0\\
C_z\Phi^{-1}+D_{zu}K\Phi^{-1} & D_{zw} & 0 & -I_{p_z}
\end{array} \right] \prec 0.
$$
Pre- and post-multiplying the last inequality by
$\blockdiag(\Phi,I_{m_w},I_{n_x},I_{p_z})\succ0$ yields
\begin{equation}\label{eq:LMI 4x4 sf}
\left[
\begin{array}{cccc}
-\Phi & 0 & A^\T+K^\T B_u^\T & C_z^\T+K^\T D_{zu}^\T\\
0 & -\eta I_{m_w} & B_w^\T & D_{zw}^\T\\
A+B_uK & B_w & -\Phi^{-1} & 0\\
C_z+D_{zu}K & D_{zw} & 0 & -I_{p_z}
\end{array} \right] \prec 0.
\end{equation}
Then, by Lemma~\ref{lemma:SANBRL inverse matrices},
from~(\ref{eq:det Psi inequality sf}), (\ref{eq:LMI 3x3 sf}),
(\ref{eq:LMI 4x4 sf}), (\ref{eq:pos def vars sf}) it follows that
the controller gain matrix $K$ is the solution to
Problem~\ref{problem:aniso sub sf} for the closed-loop
realization~(\ref{eq:CL realization sf}), which completes the
proof.
\end{proof}

\begin{remark}
Although it is not hard to prove that the synthesis
inequalities~(\ref{eq:det Psi inequality sf})--(\ref{eq:pos def
vars sf}) and the conditions~(\ref{eq:det inequality ANBRL inverse
matrices})--(\ref{eq:LMI 4x4 ANBRL inverse matrices}) of
Lemma~\ref{lemma:SANBRL inverse matrices} are equivalent, we can
only establish and prove sufficient existence conditions for the
controller~(\ref{eq:static sf controller}) since the conditions of
Lemma~\ref{lemma:SANBRL inverse matrices} are only sufficient.
This also concerns two further synthesis theorems.
\end{remark}

\begin{corollary}
The inequalities~(\ref{eq:det Psi inequality sf})--(\ref{eq:pos
def vars sf}) are not only convex in $\Psi$ and affine with
respect to $\Pi$ and $\Lambda$, but also linear in $\gamma^2.$
Obviously, minimizing $\gamma^2$ under the convex
constraints~(\ref{eq:det Psi inequality sf})--(\ref{eq:pos def
vars sf}), we minimize $\gamma$ under the same constraints. With
the notation $\wh{\gamma} := \gamma^2$, the conditions of
Theorem~\ref{theorem:state-feedback problem solution} allow to
compute the minimal $\gamma$ via solving the convex optimization
problem
\begin{equation}\label{eq:minimum gamma2 ssf}
\begin{array}{c}
\mathrm{minimize}\quad\wh{\gamma}\\
\mathrm{over}\quad \Psi, \Pi, \Lambda, \eta, \wh{\gamma}\quad
\mathrm{satisfying}\quad \mbox{(\ref{eq:det Psi inequality
sf})--(\ref{eq:pos def vars sf})}.
\end{array}
\end{equation}
If the convex problem~(\ref{eq:minimum gamma2 ssf}) is solvable,
the state-feedback controller gain matrix is constructed just as
in Theorem~\ref{theorem:state-feedback problem solution}.
\end{corollary}

All anisotropic controllers obtained from solutions to
optimization problems like~(\ref{eq:minimum gamma2 ssf}) will be
referred to as anisotropic $\gamma$-\emph{optimal} controllers.

\subsection{Fixed-order output-feedback controller design:\newline  convex constraints on reciprocal matrices}\label{subsect:output feedback}

Direct application of the sufficient conditions~(\ref{eq:det
inequality ANBRL inverse matrices})--(\ref{eq:LMI 4x4 ANBRL
inverse matrices}) of Lemma~\ref{lemma:SANBRL inverse matrices} to
the closed-loop realization
\begin{equation}\label{eq:CL realization fo of}
\left[
\begin{array}{c|c}
\EuScript{A} & \EuScript{B}\\\hline \EuScript{C} & \EuScript{D}
\end{array}
\right] = \left[
\begin{array}{cc|c}
A+B_u{D_\c} C_y & B_u{C_\c} & B_w+B_u{D_\c} D_{yw}\\
{B_\c} C_y & {A_\c} & {B_\c} D_{yw}
\\\hline
C_z+D_{zu}{D_\c} C_y & D_{zu}{C_\c} & D_{zw}+D_{zu}{D_\c} D_{yw}
\end{array}
\right]
\end{equation}
 yields the following corollary on the straightforward solution to
 general
Problem~\ref{problem:anisotropic suboptimal design}.

\begin{corollary}\label{corollay:fo of problem solution}
Given $a\geqslant 0$, $\gamma>0$, a dynamic output-feedback
controller $K$ of order $n_\xi$ defined by~(\ref{eq:dynamic
controller}) solving Problem~\ref{problem:anisotropic suboptimal
design} exists if the inequalities
\begin{equation}\label{eq:det inequality fo of}
\eta-(\det{(\e^{-2a/m_w}{\Psi})})^{1/m_w} < \gamma^2,
\end{equation}
\begin{equation}\label{eq:LMI 4x4 fo of}
\left[
\begin{array}{cccc}
\Psi-\eta I_{m_w} & \ast & \ast & \ast\\
B_w+B_uD_\c D_{yw} & -\Pi_{11} & \ast & \ast\\
B_\c D_{yw} & -\Pi_{12}^\T & -\Pi_{22} & \ast\\
D_{zw}+D_{zu}D_\c D_{yw} & 0 & 0 & -I_{p_z}
\end{array} \right] \prec 0,
\end{equation}
\begin{equation}\label{eq:LMI 6x6 fo of} \left[
\begin{array}{cccccc}
-{\Phi_{11}} & \ast & \ast & \ast & \ast & \ast\\
-{\Phi_{12}^\T} & -{\Phi_{22}} & \ast & \ast & \ast & \ast\\
0 & 0 & -\eta I_{m_w} & \ast & \ast & \ast\\
A+B_uD_\c C_y & B_uC_\c & B_w+B_uD_\c D_{yw} & -{\Pi_{11}} & \ast & \ast\\
B_\c C_y & A_\c & B_\c D_{yw} & -{\Pi_{12}^\T} & -{\Pi_{22}} & \ast\\
C_z+D_{zu}D_\c C_y & D_{zu}C_\c & D_{zw}+D_{zu}D_\c D_{yw} & 0 & 0
& -I_{p_z}
\end{array} \right] \prec 0,
\end{equation}
\begin{equation}\label{eq:pos def vars fo of}
\eta>\gamma^2,\quad \Psi\succ0,\quad\Phi := \left[
\begin{array}{cc}
\Phi_{11} & \Phi_{12}\\
\Phi_{12}^\T & \Phi_{22}
\end{array}
\right]\succ 0,\quad \Pi := \left[
\begin{array}{cc}
\Pi_{11} & \Pi_{12}\\
\Pi_{12}^\T & \Pi_{22}
\end{array}
\right]\succ 0
\end{equation}
are feasible with respect to the scalar variable $\eta$, real
$(m_w\times m_w)$-matrix $\Psi$, matrices $A_\c\in\R^{n_\xi\times
n_\xi}$, $B_\c\in\R^{n_\xi\times p_y}$, $C_\c\in\R^{m_u\times
n_\xi}$, $D_\c\in\R^{m_u\times p_y}$ and two reciprocal $(n\times
n)$-matrices $\Phi$, $\Pi$ such that
\begin{equation}\label{eq:block inverse matrices fo of}
\Phi\Pi = I_n
\end{equation}
where $n = n_x+n_\xi$ is the closed-loop system order.
\end{corollary}

Thus, the problem of finding the realization matrices
$(A_\c,B_\c,C_\c,D_\c)$ of the fixed-order output-feedback dynamic
controller~(\ref{eq:dynamic controller}) solving
Problem~\ref{problem:anisotropic suboptimal design} leads to
solving the problem~(\ref{eq:det inequality fo
of})--(\ref{eq:block inverse matrices fo of}) or making sure of
its insolvability. The problem~(\ref{eq:det inequality fo
of})--(\ref{eq:block inverse matrices fo of}) is nonconvex because
of the condition~(\ref{eq:block inverse matrices fo of}). Although
application of the known algorithms
of~\cite{IS_1995_XY_centering}--\cite{Polyak_Gryazina_2010}
 can leads to a
successful  solution of the problem~(\ref{eq:det inequality fo
of})--(\ref{eq:block inverse matrices fo of}), it should be kept
in mind that any of them can converge to local minima.
Nevertheless, the full-order controller synthesis allows for a
quite standard convexification procedure which is considered below
to be applied.

\subsection{Full-order output-feedback controller}\label{subsect:output-feedback full-order}

For full-order design ($n_x = n_\xi$) one can effectively apply
the well-known linearizing change of variables presented
in~\cite{Gahinet_1996} and used in~\cite{SGC_1997} in the
multi-objective control framework. From the block partitioning
in~(\ref{eq:pos def vars fo of}) and the condition~(\ref{eq:block
inverse matrices fo of}) it follows that
\begin{equation}\label{eq:block inverse condition}
\Phi\left[
\begin{array}{c}
\Pi_{11}\\
\Pi_{12}^\T
\end{array}
\right] = \left[
\begin{array}{c}
I_{n_x}\\
0
\end{array}
\right]
\end{equation}
which leads to
$$
\Phi\Pi_1 = \Phi_1,\qquad \Pi\Phi_1 = \Pi_1
$$
with the notation
\begin{equation}\label{eq:Phi1 Pi1 def}
\Phi_1 := \left[
\begin{array}{cc}
I_{n_x} & \Phi_{11}\\
0 & \Phi_{12}^\T
\end{array}
\right],\qquad \Pi_1 := \left[
\begin{array}{cc}
\Pi_{11} & I_{n_x}\\
\Pi_{12}^\T & 0
\end{array}
\right].
\end{equation}
It can be easily shown by direct calculation that
\begin{equation}\label{eq:congruence transformation identity}
\Pi_1^\T\Phi\Pi_1 = \Phi_1^\T\Pi_1 = \Phi_1^\T\Pi\Phi_1 =
\Pi_1^\T\Phi_1 = \left[\begin{array}{cc} \Pi_{11} & I_{n_x}\\
I_{n_x} & \Phi_{11}
\end{array}
\right].
\end{equation}
The key linearizing change of the controller variables is defined
as follows~\cite{Gahinet_1996}
\begin{eqnarray}\label{eq:euAc}
\euAc & := & \Phi_{12}A_\c\Pi_{12}^\T+\Phi_{12}B_\c
C_y\Pi_{11}+\Phi_{11}B_uC_\c\Pi_{12}^\T+\Phi_{11}(A+B_uD_\c
C_y)\Pi_{11},\\\label{eq:euBc} \euBc & := &
\Phi_{12}B_\c+\Phi_{11}B_uD_\c,\\\label{eq:euCc} \euCc & := &
C_\c\Pi_{12}^\T+D_\c C_y\Pi_{11},\\\label{eq:euDc} \euDc & := &
D_\c.
\end{eqnarray}
The new variables $\euAc$, $\euBc$, $\euCc$, $\euDc$ have
dimensions $n_x\times n_x$, $n_x\times p_y$, $m_u\times n_x$, and
$m_u\times p_y$, respectively, even if $n_x \neq n_\xi$. It is
noted in~\cite{SGC_1997} that if $\Phi_{12}$ and $\Pi_{12}$ have
full row rank and if $\euAc$, $\euBc$, $\euCc$, $\euDc$,
$\Pi_{11}$, and $\Phi_{11}$ are known, one can always find the
controller matrices $A_\c$, $B_\c$, $C_\c$, $D_\c$
satisfying~(\ref{eq:euAc})--(\ref{eq:euDc}). If the matrices
$\Phi_{12}$ and $\Pi_{12}$ are square ($n_x = n_\xi$) and
invertible, then $A_\c$, $B_\c$, $C_\c$, and $D_\c$ are unique,
i.e. for full-order design, when one can always assume that
$\Phi_{12}$ and $\Pi_{12}$ have full row rank, the mapping defined
by~(\ref{eq:euAc})--(\ref{eq:euDc}) is bijective. More details can
be found in~\cite{Gahinet_1996}, \cite{SGC_1997}.

The solution to Problem~\ref{problem:anisotropic suboptimal
design} in the full-order design case is given by
\begin{theorem}\label{theorem:full-order output-feedback problem solution}
Given $a\geqslant 0$, $\gamma>0$, a dynamic output-feedback
controller $K$ of full order $n_\xi=n_x$ defined
by~(\ref{eq:dynamic controller}) solving
Problem~\ref{problem:anisotropic suboptimal design} exists if the
convex problem
\begin{equation}\label{eq:det inequality of fullord}
\eta-(\det{(\e^{-2a/m_w}{\Psi})})^{1/m_w} < \gamma^2,
\end{equation}
\begin{equation}\label{eq:LMI 4x4 of fullord}
\left[
\begin{array}{cccc}
\Psi-\eta I_{m_w} & \ast & \ast & \ast\\
B_w+B_u\euDc D_{yw} & -\Pi_{11} & \ast & \ast\\
\Phi_{11}B_w+\euBc D_{yw} & -I_{n_x} & -\Phi_{11} & \ast\\
D_{zw}+D_{zu}\euDc D_{yw} & 0 & 0 & -I_{p_z}
\end{array} \right] \prec 0,
\end{equation}
\begin{equation}\label{eq:LMI 6x6 of fullord}
\left[
\begin{array}{cccccc}
-\Pi_{11} & \ast  & \ast & \ast & \ast & \ast\\
-I_{n_x} & -\Phi_{11} & \ast & \ast & \ast & \ast\\
0 & 0 & -\eta I_{m_w} & \ast & \ast & \ast\\
A\Pi_{11}+B_u\euCc & A+B_u\euDc C_y & B_w+B_u\euDc D_{yw} &
-\Pi_{11} & \ast & \ast\\
\euAc & \Phi_{11}A+\euBc C_y & \Phi_{11}B_w+\euBc D_{yw} &
-I_{n_x} & -\Phi_{11} & \ast\\
C_z\Pi_{11}+D_{zu}\euCc & C_z+D_{zu}\euDc C_y & D_{zw}+D_{zu}\euDc
D_{yw} & 0 & 0 & -I_{p_z}
\end{array} \right] \prec 0,
\end{equation}
\begin{equation}\label{eq:fullord arg posdef condition}
\eta>\gamma^2,\quad \Pi_{11}\succ0,\quad\Phi_{11} \succ
0,\quad\left[
\begin{array}{cc}
\Pi_{11} & I_{n_x}\\
I_{n_x} & \Phi_{11}
\end{array}
\right]\succ 0
\end{equation}
is feasible with respect to the scalar variable $\eta$, real
$(m_w\times m_w)$-matrix $\Psi$, matrices $\euAc\in\R^{n_x\times
n_x}$, $\euBc\in\R^{n_x\times p_y}$, $\euCc\in\R^{m_u\times n_x}$,
$\euDc\in\R^{m_u\times p_y}$ and two real $(n_x\times
n_x)$-matrices $\Pi_{11}$, $\Phi_{11}$. If the
problem~(\ref{eq:det inequality of fullord})--(\ref{eq:fullord arg
posdef condition}) is feasible and the unknown variables have been
found, then the controller matrices $A_\c\in\R^{n_x\times n_x}$,
$B_\c\in\R^{n_x\times p_y}$, $C_\c\in\R^{m_u\times n_x}$,
$D_\c\in\R^{m_u\times p_y}$ are uniquely defined by
\begin{eqnarray}\label{eq:Dc backtrans}
D_\c & := & \euDc,\\\label{eq:Cc backtrans} C_\c & := &
(\euCc-D_\c C_y\Pi_{11})\Pi_{12}^{-\T},\\\label{eq:Bc backtrans}
B_\c & := & \Phi_{12}^{-1}(\euBc-\Phi_{11}B_uD_\c),\\\label{eq:Ac
backtrans} A_\c & := & \Phi_{12}^{-1}(\euAc-\Phi_{12}B_\c
C_y\Pi_{11}-\Phi_{11}B_uC_\c\Pi_{12}^\T-\Phi_{11}(A+B_uD_\c
C_y)\Pi_{11})\Pi_{12}^{-\T}
\end{eqnarray}
and determined from finding two nonsingular $(n_x\times
n_x)$-matrices $\Pi_{12}$, $\Phi_{12}$ that satisfy
\begin{equation}\label{eq:block inverse condition 1}
\Pi_{12}\Phi_{12}^\T = I_{n_x}-\Pi_{11}\Phi_{11}.
\end{equation}
\end{theorem}
\begin{proof}
Let a solution to~(\ref{eq:det inequality of
fullord})--(\ref{eq:fullord arg posdef condition}) exist.
From~(\ref{eq:Phi1 Pi1 def})--(\ref{eq:euDc}) and (\ref{eq:CL
realization fo of}) it follows that
$$
\left[
\begin{array}{cc}
A\Pi_{11}+B_u\euCc & A+B_u\euDc C_y\\
\euAc & \Phi_{11}A+\euBc C_y
\end{array}
\right] = \Phi_1^\T\euA\Pi_1,\quad  \left[
\begin{array}{c}
B_w+B_u\euDc D_{yw}\\
\Phi_{11}B_w+\euBc D_{yw}
\end{array}
\right] = \Phi_1^\T\euB,
$$
$$
\left[
\begin{array}{cc}
C_z\Pi_{11}+D_{zu}\euCc & C_z+D_{zu}\euDc C_y
\end{array}
\right] = \euC\Pi_1,\quad \left[\begin{array}{cc} \Pi_{11} & I_{n_x}\\
I_{n_x} & \Phi_{11}
\end{array}
\right] = \Pi_1^\T\Phi\Pi_1 = \Phi_1^\T\Pi\Phi_1,
$$
where $\Phi$ and $\Pi$ are defined by~(\ref{eq:pos def vars fo
of}) and satisfy~(\ref{eq:block inverse matrices fo of}) with
$n_\xi=n_x$. Substitution of the above identities to the
inequalities~(\ref{eq:LMI 4x4 of fullord}), (\ref{eq:LMI 6x6 of
fullord}) yields
\begin{equation}\label{eq:ANBRL fullord LMIs}
\left[
\begin{array}{ccc}
\Psi-\eta I_{m_w} & \EuScript{B}^\T\Phi_1 & \euD^\T
\\\Phi_1^\T\EuScript{B} & -\Phi_1^\T\Pi\Phi_1 & 0\\
\euD & 0 & -I_{p_z}
\end{array}
\right]\prec 0,\quad \left[
\begin{array}{cccc}
-\Pi_1^\T{\Phi}\Pi_1 & 0 & \Pi_1^\T\EuScript{A}^\T\Phi_1 & \Pi_1^\T\EuScript{C}^\T\\
0 & -\eta I_{m_w} & \EuScript{B}^\T\Phi_1 & \EuScript{D}^\T\\
\Phi_1^\T\EuScript{A}\Pi_1 & \Phi_1^\T\EuScript{B} & -\Phi_1^\T\Pi\Phi_1 & 0\\
\EuScript{C}\Pi_1 & \EuScript{D} & 0 & -I_{p_z}
\end{array} \right] \prec 0.
\end{equation}
Performing a congruence transformation with
$$
\blockdiag(I_{m_w},\Phi_1^{-\T},I_{p_z}),\quad
\blockdiag(\Pi_1^{-\T},I_{m_w},\Phi_1^{-\T}, I_{p_z})
$$
on the inequalities~(\ref{eq:ANBRL fullord LMIs}), respectively,
leads to
\begin{equation}\label{eq:ANBRL fullord LMIs implicit}
\left[
\begin{array}{ccc}
\Psi-\eta I_{m_w} & \EuScript{B}^\T & \euD^\T
\\\EuScript{B} & -\Pi & 0\\
\euD & 0 & -I_{p_z}
\end{array}
\right]\prec 0,\quad \left[
\begin{array}{cccc}
-{\Phi} & 0 & \EuScript{A}^\T & \EuScript{C}^\T\\
0 & -\eta I_{m_w} & \EuScript{B}^\T & \EuScript{D}^\T\\
\EuScript{A} & \EuScript{B} & -\Pi & 0\\
\EuScript{C} & \EuScript{D} & 0 & -I_{p_z}
\end{array} \right] \prec 0.
\end{equation}
Then, by Lemma~\ref{lemma:SANBRL inverse matrices},
from~(\ref{eq:det inequality of fullord}), (\ref{eq:ANBRL fullord
LMIs implicit}), (\ref{eq:pos def vars fo of}), (\ref{eq:block
inverse matrices fo of}) it follows that the closed-loop
system~(\ref{eq:CL realization fo of}) is internally stable and
its $a$-anisotropic norm does not exceed the designed threshold
$\gamma$, i.e. the inequality~(\ref{eq:anisotropic suboptimality
condition}) holds. The procedure of reconstruction of the
controller realization $(A_\c,B_\c,C_\c,D_\c)$ from the solution
variables $(\euA_\c,\euB_\c,\euC_\c,\euD_\c)$ by~(\ref{eq:block
inverse condition 1}), (\ref{eq:Dc backtrans})--(\ref{eq:Ac
backtrans}) is quite standard~\cite{Gahinet_1996},
\cite{SGC_1997}.
\end{proof}
\begin{corollary}
As the inequalities~(\ref{eq:det inequality of
fullord})--(\ref{eq:fullord arg posdef condition}) are also linear
in $\wh{\gamma} := \gamma^2$, the conditions of
Theorem~\ref{theorem:full-order output-feedback problem solution}
allow to compute the minimal $\gamma$ via solving the convex
optimization problem
\begin{equation}\label{eq:minimum gamma2 fo of}
\begin{array}{c}
\mathrm{minimize}\quad\wh{\gamma}\\
\mathrm{over}\quad \Psi, \Phi_{11}, \Pi_{11}, \euAc, \euBc, \euCc,
\euDc, \eta, \wh{\gamma}\quad \mathrm{satisfying}\quad
\mbox{(\ref{eq:det inequality of fullord})--(\ref{eq:fullord arg
posdef condition})}.
\end{array}
\end{equation}
If the convex problem~(\ref{eq:minimum gamma2 fo of}) is solvable,
the controller matrices are constructed just as in
Theorem~\ref{theorem:full-order output-feedback problem solution}.
\end{corollary}
It is stressed in~\cite{SGC_1997} that the applied synthesis
procedure does not introduce any conservatism, if the analysis
result does not involve any.

The results of Theorem~\ref{theorem:full-order output-feedback
problem solution} make possible application of the anisotropic
norm as a closed-loop performance specification or objective for
specific closed-loop channels in the multi-objective control
problems based on a common Lyapunov functions~\cite{SGC_1997}
together with other performance specifications and objectives that
can be captured in the LMI framework.

\subsection{Static output-feedback controller}\label{subsect:static ouput-feedback}

Let us now consider the special and very important case of static
output-feedback controller
\begin{equation}\label{eq:static of controller}
u_k = Ky_k.
\end{equation}
\begin{problem}\label{problem:aniso sub sof}
Given LDTI plant $P$ described by~(\ref{eq:standard plant}), a
mean anisotropy level $a\geqslant 0$ of the external disturbance
$W$, and some designed threshold value $\gamma>0$, find the static
output-feedback controller~(\ref{eq:static of controller}) which
internally stabilizes the closed-loop system $T_{ZW}(z)$ with the
state-space realization
\begin{equation}\label{eq:CL realization sof}
\left[
\begin{array}{c|c}
\EuScript{A} & \EuScript{B}\\\hline \EuScript{C} & \EuScript{D}
\end{array}
\right] = \left[
\begin{array}{c|c}
A+B_uK C_y & B_w+B_uK D_{yw}\\\hline C_z+D_{zu}K C_y &
D_{zw}+D_{zu}K D_{yw}
\end{array}
\right]
\end{equation}
 and ensures its $a$-anisotropic norm
does not exceed the threshold $\gamma$, i.e.
\begin{equation}\label{eq:aninorm sub condition sof}
\sn T_{ZW}\sn_a < \gamma.
\end{equation}
\end{problem}

Direct application of the sufficient conditions~(\ref{eq:det
inequality ANBRL inverse matrices})--(\ref{eq:LMI 4x4 ANBRL
inverse matrices}) of Lemma~\ref{lemma:SANBRL inverse matrices} to
the closed-loop realization~(\ref{eq:CL realization sof})
 yields the following corollary on the straightforward solution to
Problem~\ref{problem:aniso sub sof}.

\begin{corollary}\label{corollary:sof problem solution nonconv}
Given $a\geqslant 0$, $\gamma>0$, the static output-feedback
controller~(\ref{eq:static of controller}) solving
Problem~\ref{problem:aniso sub sof} exists if the inequalities
\begin{equation}\label{eq:det inequality sof nonconv}
\eta-(\det{(\e^{-2a/m_w}{\Psi})})^{1/m_w} < \gamma^2,
\end{equation}
\begin{equation}\label{eq:LMI 3x3 sof nonconv}
\left[
\begin{array}{ccc}
\Psi-\eta I_{m_w} & \ast & \ast\\
B_w+B_uK D_{yw} & -\Pi & \ast\\
D_{zw}+D_{zu}K D_{yw} & 0 & -I_{p_z}
\end{array}
\right] \prec 0,
\end{equation}
\begin{equation}\label{eq:LMI 4x4 sof nonconv}
\left[
\begin{array}{cccc}
-\Phi & \ast & \ast & \ast\\
0 & -\eta I_{m_w} & \ast & \ast\\
A+B_u{K} C_y & B_w+B_u{K} D_{yw} & -\Pi & \ast\\
C_z+D_{zu}{K} C_y & D_{zw}+D_{zu}{K}D_{yw} & 0 & -I_{p_z}
\end{array} \right] \prec 0,
\end{equation}
\begin{equation}\label{eq:pos def vars sof nonconv}
\eta>\gamma^2,\quad \Psi\succ0,\quad \Phi\succ0,\quad \Pi\succ0
\end{equation}
are feasible with respect to the scalar variable $\eta$, real
$(m_w\times m_w)$-matrix $\Psi$, real $(m_u\times p_y)$-matrix
$K$, and two reciprocal real $(n_x\times n_x)$-matrices $\Phi$,
$\Pi$ such that
\begin{equation}\label{eq:inverse matrices sof nonconv}
\Phi\Pi = I_{n_x}.
\end{equation}
\end{corollary}
So, the problem of finding the output-feedback gain matrix $K$
solving Problem~\ref{problem:aniso sub sof} leads to solving the
problem~(\ref{eq:det inequality sof nonconv})--(\ref{eq:inverse
matrices sof nonconv}) or making sure of its insolvability. The
inequalities~(\ref{eq:det inequality sof
nonconv})--(\ref{eq:inverse matrices sof nonconv}) derived from
the straightforward application of Lemma~\ref{lemma:SANBRL inverse
matrices} are not convex because of the condition~(\ref{eq:inverse
matrices sof nonconv}). One can try to solve this general problem
by the algorithms
of~\cite{IS_1995_XY_centering}--\cite{Polyak_Gryazina_2010}
suitable for finding reciprocal matrices under convex constraints.

However, the specific linearizing change of variables presented
in~\cite{Scherer_2000} can make the resulting optimization problem
convex for a specific class of plants defined by a certain
structural property. Namely, suppose that the transfer function of
the plant~(\ref{eq:standard plant}) from the control input to
measured output vanishes, i.e.~\cite{Scherer_2000}
\begin{equation}\label{eq:Tyu vanishes}
T_{yu}(z) := C_y(zI-A)^{-1}B_u = 0.
\end{equation}
For the stabilizable and detectable plant~(\ref{eq:standard
plant}), if~(\ref{eq:Tyu vanishes}) holds, then there exists a
similarity transformation $T$ such that
\begin{equation}\label{eq:Kalman canonical decomposition sof}
\left[
\begin{array}{c|cc}
TA T^{-1} & TB_w & TB_u\\
\hline C_zT^{-1} & D_{zw} & D_{zu}\\
C_yT^{-1} & D_{yw} & 0
\end{array}
\right] = \left[
\begin{array}{cc|cc}
A_{11} & A_{12} & B_{w_1} & B_{u_1}\\
0 & A_{22} & B_{w_2} & 0\\\hline C_{z_1} & C_{z_2} & D_{zw} &
D_{zu}\\
0 & C_{y_2} & D_{yw} & 0
\end{array}
\right]
\end{equation}
where $(A_{11},B_{u_1})$ is controllable, $(A_{11},C_{y_2})$ is
observable, and the matrix $A_{22}$ is stable~\cite{Scherer_2000};
see also~\cite{Poznyak_2007}. The representation~(\ref{eq:Kalman
canonical decomposition sof}) implies that the closed-loop system
realization after static output feedback becomes
\begin{equation}\label{eq:CL realization Kalman decomp sof}
\left[
\begin{array}{c|c}
\euA & \euB\\\hline \euC & \euD
\end{array}
\right] = \left[
\begin{array}{cc|c}
A_{11} & A_{12}+B_{u_1}K C_{y_2} & B_{w_1}+B_{u_1}K D_{yw}\\
0 & A_{22} & B_{w_2}\\\hline C_{z_1} & C_{z_2}+D_{zu}K C_{y_2} &
D_{zw}+D_{zu}K D_{yw}
\end{array}
\right].
\end{equation}

The Lyapunov matrix $\Phi$ in the inequalities~(\ref{eq:LMI 3x3
ANBRL inverse matrices}), (\ref{eq:LMI 4x4 ANBRL inverse
matrices}) of Lemma~\ref{lemma:SANBRL inverse matrices} is
partitioned according to the representation of $\euA$
in~(\ref{eq:CL realization Kalman decomp sof})
as~\cite{Scherer_2000}
\begin{equation}\label{eq:Phi partition Scherer sof}
\Phi = \left[
\begin{array}{cc}
\Phi_{11} & \Phi_{12}\\
\Phi_{12}^\T & \Phi_{22}
\end{array}
\right]\succ0.
\end{equation}
The key linearizing change of variables is defined
in~\cite{Scherer_2000} as
\begin{equation}\label{eq:Scherer linearizing change sof}
P := \left[
\begin{array}{cc}
Q & S\\
S^\T & R
\end{array}
\right] = \left[
\begin{array}{cc}
\Phi_{11}^{-1} & -\Phi_{11}^{-1}\Phi_{12}\\
-\Phi_{12}^\T\Phi_{11}^{-1} &
\Phi_{22}-\Phi_{12}^\T\Phi_{11}^{-1}\Phi_{12}
\end{array}
\right].
\end{equation}
It is noted in~\cite{Scherer_2000} that the
transformation~(\ref{eq:Scherer linearizing change sof}) maps the
set of all positive definite matrices into the set of all matrices
with positive definite diagonal blocks; this map is bijective; its
inverse is given by
\begin{equation}\label{eq:Scherer back linearizing change sof}
\left[
\begin{array}{cc}
\Phi_{11} & \Phi_{12}\\
\Phi_{12}^\T & \Phi_{22}
\end{array}
\right] = \left[
\begin{array}{cc}
Q^{-1} & -Q^{-1}S\\
-S^\T Q^{-1} & R-S^\T Q^{-1}S
\end{array}
\right].
\end{equation}
The transformation~(\ref{eq:Scherer linearizing change sof}) is
motivated by the factorization~\cite{Scherer_2000}
\begin{equation}\label{eq:Scherer trans factorization}
P_1\Phi = P_2
\end{equation}
with
\begin{equation}\label{eq:Scherer congr trans fact}
P_1 := \left[
\begin{array}{cc}
Q & 0\\
S^\T & I
\end{array}
\right],\qquad P_2 := \left[
\begin{array}{cc}
I & -S\\
0 & R
\end{array}
\right].
\end{equation}

\begin{theorem}\label{theorem:sof problem solution convex}
Suppose that the plant $P$ described by~(\ref{eq:standard plant})
is such that $T_{yu}(z) = 0$, i.e.~(\ref{eq:Tyu vanishes}) holds.
Given $a\geqslant 0$, $\gamma>0$, a static output-feedback
controller defined by~(\ref{eq:static of controller}) solving
Problem~\ref{problem:aniso sub sof} exists if the convex
problem
\begin{equation}\label{eq:det inequality sof conv}
\eta-(\det{(\e^{-2a/m_w}{\Psi})})^{1/m_w} < \gamma^2,
\end{equation}
\begin{equation}\label{eq:LMI 3x3 sof conv}
\left[
\begin{array}{cccc}
\Psi-\eta I_{m_w} & \ast & \ast & \ast\\
B_{w_1}+B_{u_1}K D_{yw}-SB_{w_2} & -Q & \ast & \ast\\
RB_{w_2} & 0 & -R & \ast\\
D_{zw}+D_{zu}K D_{yw} & 0 & 0 & -I_{p_z}
\end{array} \right] \prec 0,
\end{equation}
\begin{equation}\label{eq:LMI 6x6 sof conv}
\left[
\begin{array}{cccccc}
-Q & \ast & \ast & \ast & \ast & \ast\\
0 & -R & \ast & \ast & \ast & \ast\\
0 & 0 & -\eta I_{m_w} & \ast & \ast & \ast\\
A_{11}Q & A_{11}S-SA_{22}+A_{12}+B_{u_1}K C_{y_2} & B_{w_1}+B_{u_1}K D_{yw}-SB_{w_2} & -Q & \ast & \ast\\
0 & RA_{22} & RB_{w_2} & 0 & -R & \ast\\
C_{z_1}Q & C_{z_1}S+C_{z_2}+D_{zu}K C_{y_2} & D_{zw}+D_{zu}K
D_{yw} & 0 & 0 & -I_{p_z}
\end{array} \right] \prec 0,
\end{equation}
\begin{equation}\label{eq:pos def vars sof conv}
\eta>\gamma^2,\quad \Psi\succ0,\quad Q\succ0,\quad R\succ0
\end{equation}
is feasible with respect to the scalar variable $\eta$, real
$(m_w\times m_w)$-matrix $\Psi$, controller gain matrix $K$ and
real matrices $Q$, $R$, and $S$.
\end{theorem}
\begin{proof}
Let a solution to~(\ref{eq:det inequality sof conv})--(\ref{eq:pos
def vars sof conv}) exist. Then from~(\ref{eq:Scherer congr trans
fact}), (\ref{eq:CL realization Kalman decomp sof}),
(\ref{eq:Scherer back linearizing change sof}) it follows that
\begin{eqnarray}\label{eq:transfromation identities Scherer sof 1}
\left[
\begin{array}{cc}
Q & 0\\
0 & R
\end{array}
\right] & = & P_1\Phi P_1^\T,\\
 \left[
\begin{array}{cc}
A_{11}Q & A_{11}S-SA_{22}+A_{12}+B_{u_1}K C_{y_2}\\
0 & RA_{22}
\end{array}
\right] & = &P_1\Phi\euA P_1^\T,\\  \left[
\begin{array}{c}
B_{w_1}+B_{u_1}K D_{yw}-SB_{w_2}\\
RB_{w_2}
\end{array}
\right] & = & P_1\Phi\euB,\\\label{eq:transfromation identities
Scherer sof 4} \left[
\begin{array}{cc}
C_{z_1}Q & C_{z_1}S+C_{z_2}+D_{zu}K C_{y_2}
\end{array}
\right] & = & \euC P_1^\T.
\end{eqnarray}
Substituting the identities~(\ref{eq:transfromation identities
Scherer sof 1})--(\ref{eq:transfromation identities Scherer sof
4}) to the LMIs~(\ref{eq:LMI 3x3 sof conv}), (\ref{eq:LMI 6x6 sof
conv}), we have
\begin{equation}\label{eq:LMIs 3x3 4x4 sof CL}
\left[
\begin{array}{ccc}
\Psi-\eta I_{m_w}& \euB^\T\Phi P_1^\T & \euD^\T\\
 P_1\Phi\euB & -P_1\Phi P_1^\T & 0\\
\euD & 0 & -I_{p_z}
\end{array}
\right]\prec 0,\quad \left[
\begin{array}{cccc}
-P_1\Phi P_1^\T & 0 & P_1\euA^\T\Phi P_1^\T & P_1\euC^\T\\
0 & -\eta I_{m_w} & \euB^\T\Phi P_1^\T & \euD^\T\\
P_1\Phi\euA P_1^\T & P_1\Phi\euB & -P_1\Phi P_1^\T & 0\\
\euC P_1^\T & \euD & 0 & -I_{p_z}
\end{array} \right] \prec 0.
\end{equation}
Performing a congruence transformation with
$$
\blockdiag(I_{m_w},P_1^{-1},I_{p_z}),\quad
\blockdiag(P_1^{-1},I_{m_w},P_1^{-1}, I_{p_z}),
$$
where $P_1$ is defined by~(\ref{eq:Scherer congr trans fact}), on
the inequalities~(\ref{eq:LMIs 3x3 4x4 sof CL}), respectively,
yields
\begin{equation}\label{eq:ANBRL sof BMIs}
\left[
\begin{array}{ccc}
\Psi-\eta I_{m_w} & \euB^\T\Phi & \euD^\T\\
\Phi\euB & -\Phi & 0\\
\euD & 0 & -I_{p_z}
\end{array}
\right]\prec 0,\qquad \left[
\begin{array}{cccc}
-\Phi & 0 & \euA^\T\Phi & \euC^\T\\
0 & -\eta I_{m_w} & \euB^\T\Phi & \euD^\T\\
\Phi\euA & \Phi\euB & -\Phi & 0\\
\euC & \euD & 0 & -I_{p_z}
\end{array}
\right]\prec 0.
\end{equation}
Pre- and post-multiplying the inequalities~(\ref{eq:ANBRL sof
BMIs}) by
$$
\blockdiag(I_{m_w},\Phi^{-1},I_{p_z})\succ0,\qquad\blockdiag(I_{n_x},I_{m_w},\Phi^{-1},I_{p_z})\succ0,
$$
respectively, we have
\begin{equation}\label{eq:ANBRL sof LMIs}
\left[
\begin{array}{ccc}
\Psi-\eta I_{m_w} & \euB^\T & \euD^\T\\
\euB & -\Phi^{-1} & 0\\
\euD & 0 & -I_{p_z}
\end{array}
\right]\prec 0,\qquad \left[
\begin{array}{cccc}
-\Phi & 0 & \euA^\T & \euC^\T\\
0 & -\eta I_{m_w} & \euB^\T & \euD^\T\\
\euA & \euB & -\Phi^{-1} & 0\\
\euC & \euD & 0 & -I_{p_z}
\end{array}
\right]\prec 0.
\end{equation}
Then, by Lemma~\ref{lemma:SANBRL inverse matrices},
from~(\ref{eq:det inequality sof conv}), (\ref{eq:ANBRL sof
LMIs}), (\ref{eq:pos def vars sof conv}), (\ref{eq:Phi partition
Scherer sof}) it follows that the controller gain matrix~$K$ is
the solution to Problem~\ref{problem:aniso sub sof} for the
plant~(\ref{eq:Kalman canonical decomposition sof}) and the
closed-loop system~(\ref{eq:CL realization Kalman decomp sof}),
which completes the proof.
\end{proof}
\begin{corollary}\label{corollary:sof problem convex opt}
The convex constraints~(\ref{eq:det inequality sof
conv})--(\ref{eq:pos def vars sof conv}) are also linear in
$\gamma^2$. With the notation $\wh{\gamma} := \gamma^2$, the
conditions of Theorem~\ref{theorem:sof problem solution convex}
allow for $\gamma$ to be minimized via solving the convex
optimization problem
\begin{equation}\label{eq:minimum gamma2 sof}
\begin{array}{c}
\mathrm{minimize}\quad\wh{\gamma}\\
\mathrm{over}\quad \Psi, Q, R, S, K, \eta, \wh{\gamma}\quad
\mathrm{satisfying}\quad \mbox{(\ref{eq:det inequality sof
conv})--(\ref{eq:pos def vars sof conv})}.
\end{array}
\end{equation}
\end{corollary}

The controller gain matrix $K$ enters the synthesis
LMIs~(\ref{eq:LMI 3x3 sof conv}), (\ref{eq:LMI 6x6 sof conv})
directly. It is noted in~\cite{Scherer_2000} that this
 allows for some
structural requirements on this controller gain to be incorporated
making possible even the synthesis of decentralized controllers
(with block-diagonal $K$) via convex optimization.

The results of Theorem~\ref{theorem:sof problem solution convex}
make possible application of the anisotropic norm as a closed-loop
performance specification or objective for specific closed-loop
channels in the multi-objective control problems with LMI
specifications considered in~\cite{Scherer_2000}.

It should be also noted that in general case, when the structural
property~(\ref{eq:Tyu vanishes}) does not hold, one can follow the
way of~\cite{Scherer_2000} and make use of the Youla-Ku\v{c}era
parametrization of stabilizing
controller~\cite{Kucera_1975,YJB_1976} to parametrize affinely the
closed-loop system, enforce the said property, and bring the
closed-loop realization to the form~(\ref{eq:Kalman canonical
decomposition sof}). Then the synthesis of the anisotropic
controller can be treated as finding the Youla parameter that
enters the closed-loop system affinely by applying the results of
Theorem~\ref{theorem:sof problem solution convex} and
Corollary~\ref{corollary:sof problem convex opt}.

Besides the class of systems which satisfy the structural
property~(\ref{eq:Tyu vanishes}),  there are two particular cases
of the system's structure which allow for the static
output-feedback design problem to lead to some convex optimization
problem by applying a nonsingular state coordinate transformation
and introducing structured slack variables just as it was done for
$\Hinf$ synthesis problem in~\cite{LLK_2006}. These cases are the
so called singular control and filtering problems.

Let us first consider the singular control problem when the matrix
$D_{zu}$ of the plant~(\ref{eq:standard plant}) is zero and the
matrix $B_u$ is of full column rank. Then there exists a
nonsingular state coordinate transformation matrix $T_u$ such
that~\cite{LLK_2006}
\begin{equation}\label{eq:TuBu}
\bB_u := T_uB_u = \left[ \begin{array}{c} I_{m_u}\\0
\end{array} \right].
\end{equation}
Under this transformation, the plant realization matrices become
\begin{equation}\label{eq:realization under Tu}
\bA := T_uAT_u^{-1},\quad \bB_w := T_uB_w,\quad \bC_z :=
C_zT_u^{-1},\quad \bC_y := C_yT_u^{-1}.
\end{equation}

\begin{theorem}\label{theorem:singular control sof convex}
Suppose that the plant $P$ described by~(\ref{eq:standard plant})
is such that $D_{zu} = 0$ and $\rank{B_u} = m_u$. Given
$a\geqslant 0$, $\gamma>0$, a static output-feedback controller
defined by~(\ref{eq:static of controller}) solving
Problem~\ref{problem:aniso sub sof} for the closed-loop
realization
\begin{equation}\label{eq:CL realization sof singular control}
\left[
\begin{array}{c|c}
\EuScript{A} & \EuScript{B}\\\hline \EuScript{C} & \EuScript{D}
\end{array}
\right] = \left[
\begin{array}{c|c}
A+B_uK C_y & B_w+B_uK D_{yw}\\\hline C_z & D_{zw}
\end{array}
\right]
\end{equation}
exists if the convex problem
\begin{equation}\label{eq:det inequality sof conv sc}
\eta-(\det{(\e^{-2a/m_w}{\Psi})})^{1/m_w} < \gamma^2,
\end{equation}
\begin{equation}\label{eq:LMI 3x3 sof conv sc}
\left[
\begin{array}{ccc}
\Psi-\eta I_{m_w} & \bB_{w}^\T\bS^\T+D_{yw}^\T L^\T & D_{zw}^\T\\
\bS\bB_{w}+LD_{yw} & \bPhi-\bS-\bS^\T & 0\\
D_{zw} & 0 & -I_{p_z}
\end{array} \right] \prec 0,
\end{equation}
\begin{equation}\label{eq:LMI 4x4 sof conv sc}
\left[
\begin{array}{cccc}
-\bPhi & 0 & \bA^\T\bS^\T+\bC_y^\T L^\T & \bC_z^\T\\
0 & -\eta I_{m_w} & \bB_w^\T\bS^\T+D_{yw}^\T L^\T & D_{zw}^\T\\
\bS\bA+L\bC_y & \bS\bB_w+LD_{yw} & \bPhi-\bS-\bS^\T & 0\\
\bC_z & D_{zw} & 0 & -I_{p_z}
\end{array} \right] \prec 0,
\end{equation}
\begin{equation}\label{eq:pos def vars sof conv sc}
\eta>\gamma^2,\quad \Psi\succ0,\quad \bPhi\succ0,
\end{equation}
where $\bA$, $\bB_w$, $\bC_z$, $\bC_y$ are defined
by~(\ref{eq:realization under Tu}), is feasible with respect to
the scalar variable $\eta$, real $(m_w\times m_w)$-matrix $\Psi$,
$(n_x\times n_x)$-matrix $\bPhi$, and two structured matrix
variables
\begin{equation}\label{eq:S L def}
\bS := \left[
\begin{array}{cc}
\bS_1 & 0\\
0 & \bS_2
\end{array}
\right],\quad L := \left[
\begin{array}{c}
L_1\\
0
\end{array}
\right].
\end{equation}
If the problem~(\ref{eq:det inequality sof conv sc})--(\ref{eq:pos
def vars sof conv sc}) is feasible and the unknown variables have
been found, then the output-feedback controller gain matrix is
determined by $K = \bS_1^{-1}L_1$.
\end{theorem}
The proof is similar to that of~\cite{LLK_2006} where it is
derived for the $\Hinf$ norm performance criterion.
\begin{proof}
Let a solution to the problem~(\ref{eq:det inequality sof conv
sc})--(\ref{eq:pos def vars sof conv sc}) exist. Performing a
congruence transformation with
$$
\blockdiag{(I_{m_w},T_u^\T,I_{p_z})},\quad\blockdiag{(T_u^\T,I_{m_w},T_u^\T,I_{p_z})}
$$
on the inequalities~(\ref{eq:LMI 3x3 sof conv sc}), (\ref{eq:LMI
4x4 sof conv sc}), respectively, leads to
$$
\left[
\begin{array}{ccc}
\Psi-\eta I_{m_w} & B_{w}^\T T_u^\T\bS^\T T_u+D_{yw}^\T L^\T T_u & D_{zw}^\T\\
T_u^\T\bS T_uB_{w}+T_u^\T LD_{yw} & T_u^\T(\bPhi-\bS-\bS^\T)T_u & 0\\
D_{zw} & 0 & -I_{p_z}
\end{array} \right] \prec 0,
$$
$$
\left[
\begin{array}{cccc}
-T_u^\T\bPhi T_u & 0 & A^\T T_u^\T\bS^\T T_u+C_y^\T L^\T T_u & C_z^\T\\
0 & -\eta I_{m_w} & B_w^\T T_u^\T\bS^\T T_u+D_{yw}^\T L^\T T_u & D_{zw}^\T\\
T_u^\T\bS T_u A+T_u^\T L C_y & T_u^\T\bS T_uB_w+T_u^\T LD_{yw} & T_u^\T(\bPhi-\bS-\bS^\T)T_u & 0\\
C_z & D_{zw} & 0 & -I_{p_z}
\end{array} \right] \prec 0
$$
where the plant realization matrices are derived from the backward
transformation of~(\ref{eq:realization under Tu}). Let us denote
$S:=T_u^\T\bS T_u$, $\Phi:=T_u^\T\bPhi T_u$. Then from~(\ref{eq:S
L def}) and definition of $K = \bS_1^{-1}L_1$  it follows that
$$
T_u^\T L = T_u^\T\left[
\begin{array}{c}
L_1\\0
\end{array}
\right] = T_u^\T\left[
\begin{array}{cc}
\bS_1 & 0\\
0 & \bS_2
\end{array}
\right]\left[
\begin{array}{c}
I_{m_u}\\0
\end{array}
\right]K = T_u^\T\bS\bB_u K = SB_uK,
$$
and the above LMIs can be rewritten as
$$
\left[
\begin{array}{ccc}
\Psi-\eta I_{m_w} & (B_{w}+B_uKD_{yw})^\T S^\T & D_{zw}^\T\\
S(B_{w}+B_uKD_{yw}) & \Phi-S-S^\T & 0\\
D_{zw} & 0 & -I_{p_z}
\end{array} \right] \prec 0,
$$
$$
\left[
\begin{array}{cccc}
-\Phi & 0 & (A+B_uKC_y)^\T S^\T & C_z^\T\\
0 & -\eta I_{m_w} & (B_{w}+B_uKD_{yw})^\T S^\T & D_{zw}^\T\\
S(A+B_uKC_y) & S(B_{w}+B_uKD_{yw}) & \Phi-S-S^\T & 0\\
C_z & D_{zw} & 0 & -I_{p_z}
\end{array} \right] \prec 0,
$$
or, in terms of the closed-loop realization~(\ref{eq:CL
realization sof singular control}), as
$$
\left[
\begin{array}{ccc}
\Psi-\eta I_{m_w} & \euB^\T S^\T & \euD^\T\\
S\euB & \Phi-S-S^\T & 0\\
\euD & 0 & -I_{p_z}
\end{array} \right] \prec 0,\quad
\left[
\begin{array}{cccc}
-\Phi & 0 & \euA^\T S^\T & \euC^\T\\
0 & -\eta I_{m_w} & \euB^\T S^\T & \euD^\T\\
S\euA & S\euB & \Phi-S-S^\T & 0\\
\euC & \euD & 0 & -I_{p_z}
\end{array} \right] \prec 0.
$$
Then, performing a congruence transformation with
$$
\blockdiag{(I_{m_w},S^{-1},I_{p_z})},\quad\blockdiag{(I_{n_x},I_{m_w},S^{-1},I_{p_z})}
$$
on the last inequalities, respectively, we have
\begin{equation}\label{eq:aninorm LMI 3x3 sof sc slack}
\left[
\begin{array}{ccc}
\Psi-\eta I_{m_w} & \euB^\T & \euD^\T\\
\euB & S^{-1}\Phi S^{-\T}-S^{-1}-S^{-\T} & 0\\
\euD & 0 & -I_{p_z}
\end{array} \right] \prec 0,
\end{equation}
\begin{equation}\label{eq:aninorm LMI 4x4 sof sc slack}
\left[
\begin{array}{cccc}
-\Phi & 0 & \euA^\T & \euC^\T\\
0 & -\eta I_{m_w} & \euB^\T & \euD^\T\\
\euA & \euB & S^{-1}\Phi S^{-\T}-S^{-1}-S^{-\T} & 0\\
\euC & \euD & 0 & -I_{p_z}
\end{array} \right] \prec 0.
\end{equation}
From the inequality
$$
(S^{-1}-\Phi^{-1})(-\Phi)(S^{-1}-\Phi^{-1})^\T \prec 0
$$
it is clear that
$$
-\Phi^{-1} \prec S^{-1}\Phi S^{-\T}-S^{-1}-S^{-\T}.
$$
Then, by Lemma~\ref{lemma:SANBRL inverse matrices},
from~(\ref{eq:det inequality sof conv sc}), (\ref{eq:aninorm LMI
3x3 sof sc slack}), (\ref{eq:aninorm LMI 4x4 sof sc slack}),
(\ref{eq:pos def vars sof conv sc}) it follows that the controller
gain matrix $K$ is the solution to Problem~\ref{problem:aniso sub
sof} for the closed-loop realization~(\ref{eq:CL realization sof
singular control}), which completes the proof.
\end{proof}

\begin{remark}
Unlike the proofs of Theorems~\ref{theorem:state-feedback problem
solution}--\ref{theorem:sof problem solution convex}, there is no
equivalence between the synthesis inequalities~(\ref{eq:det
inequality sof conv sc})--(\ref{eq:pos def vars sof conv sc}) and
the conditions~(\ref{eq:det inequality ANBRL inverse
matrices})--(\ref{eq:LMI 4x4 ANBRL inverse matrices}) of
Lemma~\ref{lemma:SANBRL inverse matrices}. The synthesis
LMIs~(\ref{eq:LMI 3x3 sof conv sc}), (\ref{eq:LMI 4x4 sof conv
sc}) establich only sufficient conditions for the
inequalities~(\ref{eq:LMI 3x3 ANBRL inverse matrices}),
(\ref{eq:LMI 4x4 ANBRL inverse matrices}) of
Lemma~\ref{lemma:SANBRL inverse matrices} to be solvable. This
also concerns a synthesis theorem below.
\end{remark}

\begin{corollary}\label{corollary:sof problem convex opt sc}
With the notation $\wh{\gamma} := \gamma^2$, the conditions of
Theorem~\ref{theorem:singular control sof convex} allow for
$\gamma$ to be minimized via solving the convex optimization
problem
\begin{equation}\label{eq:minimum gamma2 sof sc}
\begin{array}{c}
\mathrm{minimize}\quad\wh{\gamma}\\
\mathrm{over}\quad \Psi, \bPhi, \bS, L, \eta, \wh{\gamma}\quad
\mathrm{satisfying}\quad \mbox{(\ref{eq:det inequality sof conv
sc})--(\ref{eq:pos def vars sof conv sc})}.
\end{array}
\end{equation}
If the problem~(\ref{eq:minimum gamma2 sof sc}) is solvable, the
controller gain matrix is constructed just as in
Theorem~\ref{theorem:singular control sof convex}.
\end{corollary}

Now consider the singular filtering problem when the matrix
$D_{yw}$ of the plant~(\ref{eq:standard plant}) is zero and the
matrix $C_y$ is of full row rank. Then there exists a nonsingular
state coordinate transformation matrix $T_y$ such
that~\cite{LLK_2006}
\begin{equation}\label{eq:CyTy}
\bC_y := C_yT_y^{-1} = \left[ \begin{array}{cc} I_{p_y} & 0
\end{array} \right].
\end{equation}
Under this transformation, the plant realization matrices become
\begin{equation}\label{eq:realization under Ty}
\bA := T_yAT_y^{-1},\quad \bB_w := T_yB_w,\quad \bB_u :=
T_yB_u,\quad \bC_z := C_zT_y^{-1}.
\end{equation}

\begin{theorem}\label{theorem:singular filtering sof convex}
Suppose that the plant $P$ described by~(\ref{eq:standard
plant}) is such that $D_{yw} = 0$ and $\rank{C_y} = p_y$. Given
$a\geqslant 0$, $\gamma>0$, a static output-feedback controller
defined by~(\ref{eq:static of controller}) solving
Problem~\ref{problem:aniso sub sof} for the closed-loop
realization
\begin{equation}\label{eq:CL realization sof singular filtering}
\left[
\begin{array}{c|c}
\EuScript{A} & \EuScript{B}\\\hline \EuScript{C} & \EuScript{D}
\end{array}
\right] = \left[
\begin{array}{c|c}
A+B_uK C_y & B_w\\\hline C_z+D_{zu}KC_y & D_{zw}
\end{array}
\right]
\end{equation}
exists if the convex problem
\begin{equation}\label{eq:det inequality sof conv sf}
\eta-(\det{(\e^{-2a/m_w}{\Psi})})^{1/m_w} < \gamma^2,
\end{equation}
\begin{equation}\label{eq:LMI 3x3 sof conv sf}
\left[
\begin{array}{ccc}
\Psi-\eta I_{m_w} & \bB_{w}^\T & D_{zw}^\T\\
\bB_{w} & -\bPi & 0\\
D_{zw} & 0 & -I_{p_z}
\end{array} \right] \prec 0,
\end{equation}
\begin{equation}\label{eq:LMI 4x4 sof conv sf}
\left[
\begin{array}{cccc}
\bPi - \bR - \bR^\T & 0 & \bR^\T\bA^\T+M^\T\bB_u^\T & \bR^\T\bC_z^\T + M^\T D_{zu}^\T\\
0 & -\eta I_{m_w} & \bB_w^\T & D_{zw}^\T\\
\bA\bR+\bB_uM & \bB_w & -\bPi & 0\\
\bC_z\bR+D_{zu}M & D_{zw} & 0 & -I_{p_z}
\end{array} \right] \prec 0,
\end{equation}
\begin{equation}\label{eq:pos def vars sof conv sf}
\eta>\gamma^2,\quad \Psi\succ0,\quad \bPi\succ0,
\end{equation}
where $\bA$, $\bB_w$, $\bC_z$, $\bC_y$ are defined
by~(\ref{eq:realization under Ty}), is feasible with respect to
the scalar variable $\eta$, real $(m_w\times m_w)$-matrix $\Psi$,
$(n_x\times n_x)$-matrix $\bPi$, and two structured matrix
variables
\begin{equation}\label{eq:R M def}
\bR := \left[
\begin{array}{cc}
\bR_1 & 0\\
0 & \bR_2
\end{array}
\right],\quad M := \left[
\begin{array}{cc}
M_1 & 0
\end{array}
\right].
\end{equation}
If the problem~(\ref{eq:det inequality sof conv sf})--(\ref{eq:pos
def vars sof conv sf}) is feasible and the unknown variables have
been found, then the output-feedback controller gain matrix is
determined by $K = M_1\bR_1^{-1}$.
\end{theorem}
The proof is dual to that of Theorem~\ref{theorem:singular control
sof convex} and similar to that of~\cite{LLK_2006} where it is
derived for the $\Hinf$ norm performance criterion.
\begin{proof}
Let a solution to the problem~(\ref{eq:det inequality sof conv
sf})--(\ref{eq:pos def vars sof conv sf}) exist. Substitute the
realization matrices defined by~(\ref{eq:realization under Ty}) to
the LMIs~(\ref{eq:LMI 3x3 sof conv sf}), (\ref{eq:LMI 4x4 sof conv
sf}). Perform a congruence transformation with
$$
\blockdiag{(I_{m_w},T_y^{-1},I_{p_z})},\quad\blockdiag{(T_y^{-1},I_{m_w},T_y^{-1},I_{p_z})}
$$
on the LMIs~(\ref{eq:LMI 3x3 sof conv sf}), (\ref{eq:LMI 4x4 sof
conv sf}), respectively. Then define $R := T_y^{-1}\bR T_y^{-\T}$
and $\Pi := T_y^{-1}\bPi T_y^{-\T}$. From~(\ref{eq:R M def}) and
definition of $K = M_1R_1^{-1}$  it follows that $MT_y^{-\T} =
KC_yR,$ and the LMIs~(\ref{eq:LMI 3x3 sof conv sf}), (\ref{eq:LMI
4x4 sof conv sf}) can be rewritten as
$$
\left[
\begin{array}{ccc}
\Psi-\eta I_{m_w} & B_{w}^\T & D_{zw}^\T\\
B_{w} & -\Pi & 0\\
D_{zw} & 0 & -I_{p_z}
\end{array} \right] \prec 0,
$$
$$
\left[
\begin{array}{cccc}
\Pi-R-R^\T & 0 & R^\T(A+B_uKC_y)^\T & R^\T(C_z+D_{zu}KC_y)^\T\\
0 & -\eta I_{m_w} & B_w^\T & D_{zw}^\T\\
(A+B_uKC_y)R & B_w & -\Pi & 0\\
(C_z+D_{zu}KC_y)R & D_{zw} & 0 & -I_{p_z}
\end{array} \right] \prec 0,
$$
or, in terms of the closed-loop realization~(\ref{eq:CL
realization sof singular filtering}), as
\begin{equation}\label{eq:aninorm LMIs sof sf slack}
\left[
\begin{array}{ccc}
\Psi-\eta I_{m_w} & \euB^\T & \euD^\T\\
\euB & -\Pi & 0\\
\euD & 0 & -I_{p_z}
\end{array} \right] \prec 0,\quad
\left[
\begin{array}{cccc}
\Pi-R-R^\T & 0 & R^\T\euA^\T & R^\T\euC^\T\\
0 & -\eta I_{m_w} & \euB^\T & \euD^\T\\
\euA R & \euB & -\Pi & 0\\
\euC R & \euD & 0 & -I_{p_z}
\end{array} \right] \prec 0.
\end{equation}
Then, performing a congruence transformation with
$\blockdiag{(R^{-\T},I_{m_w},I_{n_x},I_{p_z})}$ on the last
inequality leads to
\begin{equation}\label{eq:aninorm LMI 4x4 sof sf slack}
\left[
\begin{array}{cccc}
R^{-\T}\Pi R^{-1}-R^{-1}-R^{-\T} & 0 & \euA^\T & \euC^\T\\
0 & -\eta I_{m_w} & \euB^\T & \euD^\T\\
\euA & \euB & -\Pi & 0\\
\euC & \euD & 0 & -I_{p_z}
\end{array} \right] \prec 0.
\end{equation}
From the inequality
$$
(R^{-1}-\Pi^{-1})^\T(-\Pi)(R^{-1}-\Pi^{-1}) \prec 0
$$
it is clear that
$$
-\Pi^{-1} \prec R^{-\T}\Pi R^{-1}-R^{-1}-R^{-\T}.
$$
Let us define $\Phi := \Pi^{-1}$. Then, by Lemma~\ref{lemma:SANBRL
inverse matrices}, from~(\ref{eq:det inequality sof conv sf}),
(\ref{eq:aninorm LMIs sof sf slack}), (\ref{eq:aninorm LMI 4x4 sof
sf slack}), (\ref{eq:pos def vars sof conv sf}) it follows that
the controller gain matrix $K$ is the solution to
Problem~\ref{problem:aniso sub sof} for the closed-loop
realization~(\ref{eq:CL realization sof singular filtering}),
which completes the proof.
\end{proof}

\begin{corollary}\label{corollary:sof problem convex opt sf}
With the notation $\wh{\gamma} := \gamma^2$, the conditions of
Theorem~\ref{theorem:singular filtering sof convex} allow for
$\gamma$ to be minimized via solving the convex optimization
problem
\begin{equation}\label{eq:minimum gamma2 sof sf}
\begin{array}{c}
\mathrm{minimize}\quad\wh{\gamma}\\
\mathrm{over}\quad \Psi, \bPi, \bR, M,\eta, \wh{\gamma}\quad
\mathrm{satisfying}\quad \mbox{(\ref{eq:det inequality sof conv
sf})--(\ref{eq:pos def vars sof conv sf})}.
\end{array}
\end{equation}
If the problem~(\ref{eq:minimum gamma2 sof sf}) is solvable, the
controller gain matrix is constructed just as in
Theorem~\ref{theorem:singular filtering sof convex}.
\end{corollary}

It is noted in~\cite{LLK_2006} that since the singular control and
filtering problems are dual, the convex feasibility
problems~(\ref{eq:det inequality sof conv sc})--(\ref{eq:pos def
vars sof conv sc}) and (\ref{eq:det inequality sof conv
sf})--(\ref{eq:pos def vars sof conv sf}) of
Theorems~\ref{theorem:singular control sof convex} and
\ref{theorem:singular filtering sof convex} are in a sense dual
too, just as the convex optimization problems~(\ref{eq:minimum
gamma2 sof sc}) and (\ref{eq:minimum gamma2 sof sf}) of
Corollaries~\ref{corollary:sof problem convex opt sc} and
\ref{corollary:sof problem convex opt sf}. Replacing the
realization matrices and the variables  under the coordinate
transformation in the formulas of Theorem~\ref{theorem:singular
control sof convex} and Corollary~\ref{corollary:sof problem
convex opt sc} as
$$
\begin{array}{rcl}
\{ \bA,\bB_w,\bB_u,\bC_z,D_{zw},\bC_y,D_{yw} \} & \longrightarrow
& \{
\bA^\T,\bC_z^\T,\bC_y^\T,\bB_w^\T,D_{zw}^\T,\bB_u^\T,D_{zu}^\T
\},\\
\{ \eta,\Psi,\bPhi,\bS,L \} & \longrightarrow & \{
\eta,\Psi,\bPi,\bR^\T,M^\T \},
\end{array}
$$
we obtain the respective formulas of Theorem~\ref{theorem:singular
filtering sof convex} and Corollary~\ref{corollary:sof problem
convex opt sf} and the controller is given by $K \longrightarrow
K^\T.$

It is shown in~\cite{LLK_2006} that the results of
Theorem~\ref{theorem:singular control sof convex} and
Corollary~\ref{corollary:sof problem convex opt sc} can be applied
to synthesis of decentralized anisotropic suboptimal and
$\gamma$-optimal static output-feedback and fixed-order
controllers. In turn, Theorem~\ref{theorem:singular filtering sof
convex} and Corollary~\ref{corollary:sof problem convex opt sf}
allow to get a solution to simultaneous anisotropic
output-feedback control problems. These topics are beyond the
limits of this paper and may be discussed elsewhere.

\subsection{Fixed-order controller via convex optimization}

It is well-known (see e.g.~\cite{IS_1994}) that the fixed-order
dynamic controller synthesis problem can be embedded into a static
output-feedback design problem by augmentation of the plant states
with the controller states as
\begin{equation}\label{eq:plant augmentation sof}
\left[
\begin{array}{ccc}
\calA & \calB_w & \calB_u\\ \calC_z & \calD_{zw} & \calD_{zu}\\
\calC_y & \calD_{yw} & 0
\end{array}
\right] := \left[
\begin{array}{cc|c|cc}
A & 0 & B_w & 0 & B_u\\
0 & 0 & 0 & I_{n_\xi} & 0\\\hline C_z & 0 & D_{zw} & 0 &
D_{zu}\\\hline 0 & I_{n_\xi} & 0 & 0 & 0\\ C_y & 0 & D_{yw} & 0 &
0
\end{array}
\right].
\end{equation}
The closed-loop realization is then given by
$$
\left[
\begin{array}{cc}
\euA & \euB\\
\euC & \euD
\end{array}
\right] = \left[
\begin{array}{cc}
\calA & \calB_w\\
\calC_z & \calD_{zw}
\end{array}
\right] + \left[
\begin{array}{c}
\calB_u\\
\calD_{zu}
\end{array}
\right]K\left[
\begin{array}{cc}
\calC_y & \calD_{yw}
\end{array}
\right] = \left[
\begin{array}{cc}
\calA+\calB_uK\calC_y & \calB_w+\calB_uK\calD_{zw}\\
\calC_z+\calD_{zw}K\calC_y & \calD_{zw}+\calD_{zu}K\calD_{yw}
\end{array}
\right]
$$
where the gain matrix $K$ incorporates the controller parameters
\begin{equation}\label{eq:augmented Theta def}
K := \left[
\begin{array}{cc}
A_\c & B_\c\\ C_\c & D_\c
\end{array}
\right].
\end{equation}
Therefore, if the realization of the plant~(\ref{eq:standard
plant}) has one of the matrices $D_{zu}$ or $D_{yw}$ identically
zero with $B_u$ or $C_y$ of full column/row rank, respectively,
 we can make use of
Theorem~\ref{theorem:singular control sof convex} and
Corollary~\ref{corollary:sof problem convex opt sc} or
Theorem~\ref{theorem:singular filtering sof convex} and
Corollary~\ref{corollary:sof problem convex opt sf} to find the
fixed-order anisotropic $\gamma$-optimal (suboptimal) controller
as the static output-feedback gain~(\ref{eq:augmented Theta def})
for the realization~(\ref{eq:plant augmentation sof}) of the
augmented plant.

\section{Numerical examples}\label{sect:numerical example}

In this section we provide several purely illustrative numerical
examples of the anisotropic $\gamma$-optimal controller design via
convex optimization. Only two special design cases are considered,
namely, the full-order output-feedback controller and static
output-feedback gain defined in Theorems~\ref{theorem:full-order
output-feedback problem solution} and \ref{theorem:singular
filtering sof convex}, respectively. As regards general
Problems~\ref{problem:anisotropic suboptimal design},
\ref{problem:aniso sub sof} of the anisotropic suboptimal
controller design with the solutions defined by
Corollaries~\ref{corollay:fo of problem solution},
\ref{corollary:sof problem solution nonconv}, testing and
benchmark of various algorithms for finding reciprocal matrices
under convex constraints
(e.g.,~\cite{El_Ghaoui_1997,Balandin_Kogan_2005,Polyak_Gryazina_2011})
is the issue of future work and will be presented elsewhere.
However, it should be mentioned that the algorithms
of~\cite{El_Ghaoui_1997,Balandin_Kogan_2005} have been tested on
some reasonable number of state-space realizations randomly
generated by the MATLAB Control Systems Toolbox function
\texttt{drss} and some models from the \emph{COMP$l_eib$}
collection~\cite{Leibfritz_2004,Leibfritz_Lipinski_2003}. The
numerical experiments have shown that application of both of that
algorithms often leads to convergence to local minima and depend
on initial conditions. The randomized technique proposed
in~\cite{Polyak_Gryazina_2011} aimed at generation of the initial
conditions seems to be able to improve the situation.

All computations have been carried out by means of MATLAB 7.9.0
(R2009b), Control System Toolbox, and Robust Control Toolbox in
combination with the YALMIP interface~\cite{Lofberg_2004} and the
SeDuMi solver~\cite{Sturm_1999} with CPU P8700 $2\times2.53$GHz.

\subsection{Full-order output-feedback design}

\subsubsection{TU-154 aircraft landing}\label{subsubsect:TU-154}

First we consider the problem of longitudinal flight control in
landing approach under the influence of both deterministic and
stochastic external disturbances in conditions of a windshear and
noisy measurements. The control aims at disturbance attenuation
and stabilization of the aircraft longitudinal motion along some
desired glidepath. The linearized discrete time-invariant model of
TU-154 aircraft landing is given in~\cite{KPTV_2004}, where the
problem was solved by means of the anisotropic optimal controller
derived in~\cite{VKS_1996_2}. Here we present the results of
solving the anisotropic $\gamma$-optimal full-order synthesis
problem via convex optimization as defined in
Theorem~\ref{theorem:full-order output-feedback problem solution}
and Corollary~\ref{corollay:fo of problem solution}.

The mathematical model of the aircraft longitudinal motion
defining deviation from a nominal trajectory was derived
in~\cite{KPTV_2004} at the trajectory point characterized by the
airspeed $V_0 = 71.375$ m/sec, flying path slope angle $\theta_0 =
-2.7$ deg, pitch angle rate $\omega_{z0} = 0$ deg/sec, pitch angle
$\vartheta_0 = 0$ deg, height $h_0 = 600$ m, and thrust $T_0 =
52540$ newton. The model has order $n_x = 6$, two control inputs
(the signal $\Delta\vartheta_{cy}$ generated by the controller to
deflect the generalized ailerons and the throttle lever position
$\Delta\delta_{t}$) and two measured outputs (the airspeed $\Delta
V+w_{V,k}$ and the height $\Delta h+w_{h,k}$). The sampling time
of the model $\Delta t = 0.01$ sec.

The anisotropic $\gamma$-optimal controller $K_a$ was derived from
a solution to the convex optimization problem~(\ref{eq:minimum
gamma2 fo of}) as defined in Theorem~\ref{theorem:full-order
output-feedback problem solution}. The state-space realization of
the anisotropic $\gamma$-optimal controller $K_a$ computed for the
mean anisotropy level $a=0.7$ is presented below together with the
realizations of $\H2$ and $\Hinf$ optimal controllers $K_2$ and
$K_\infty$ computed by MATLAB Robust Control Toolbox functions
\texttt{h2syn} (Riccati equations technique) and \texttt{hinfsyn}
(LMI optimization technique):
\begin{equation*}\label{eq:TU154 H2cont}
K_2 = \left[\scriptsize
\begin{array}{cccccc|cc}
0.9901   &   -0.0008      &      0   &   -0.0009  &  -0.000133    &   0.0009  & 0.009301  &  0.000133\\
0.002025   &    0.9962  &   0.001999  &   0.008616 &   -0.002482 &  6.243\cdot 10^{-5} & 0.0009729  &  0.003669\\
-0.007851  &   -0.01844    &   0.9754   &   -0.0292   &  -0.01198  & -0.0006086 & 0.0001711 &  0.0003985\\
-0.0001271  & -0.0002021  &   0.009825   &    0.9998  &  -0.001113 & -5.202\cdot 10^{-6} & 6.059\cdot 10^{-5} &   0.001014\\
-0.0006442   &    0.0124    &        0     &       0   &    0.9862    &        0 & 0.0001442   &  0.01381\\
-0.003035 &  -0.0006769 &  -0.0001388 &  -0.0001432  & -0.0004761
& 0.9954 & 0   &  0\\\hline -0.6649  &  -2.021  &  -0.749  &   -1.18 &  -0.9897 & -0.05202 & 0 & 0\\
-0.7587 &  -0.1692 & -0.03469 & -0.03581  &  -0.119  & -0.1572 & 0
&  0
\end{array}
\right],
\end{equation*}
\begin{equation*}\label{eq:TU154 anicont}
K_a = \left[\scriptsize
\begin{array}{cccccc|cc}
0.9959 &  -0.0001701 &  -0.0009358 &  -0.001023 &   0.005572 &
0.01975 & 1.698  &  0.7115\\
-0.001248  &    0.9946  &  0.007195 & -0.0001598  &  -0.00115 &
0.002974 & 0.2287 &   0.1535\\
0.003114  &  -0.01651   &   0.9865 &  0.0004621   &  0.01104 &
-0.004157 & -0.02124 & -0.06646\\
0.0009071 &  0.0004571 &  -0.002899   &   0.9953  &  -0.00819 &
0.006493 & -5.48 &   -1.223\\
-0.001239  & -0.004594 &  -0.002913 & -0.0006268   &   0.9905 &
-0.003993 & 22.7  &   1.848\\
-0.0006717  &   -0.0216  &   -0.0315  &  -0.06266  &  0.007809 &
0.9647 & 79.44   &  34.71\\\hline 5.558\cdot 10^{-6} &  4.835\cdot
10^{-5} & -2.522\cdot 10^{-5} & 0.0006601  &   0.002001 &
-0.001502 & -0.08091 & -0.05013\\
1.122\cdot 10^{-5} &  4.891\cdot 10^{-5}  & 7.805\cdot 10^{-5}  &
0.0001031 & -0.0005779 & -0.001997 & -0.1794 &  -0.07335
\end{array}
\right],
\end{equation*}
\begin{equation*}\label{eq:TU154 Hinfcont}
K_\infty = \left[\scriptsize
\begin{array}{cccccc|cc}
0.9953  &  -0.01065 &  0.0009194  &  0.001878   & 0.001211   &  -0.0018 & 0.005824 &  0.002953\\
0.01093   &   0.9835  &  0.001251  &  0.001524 &   0.001613 &  -0.003523 & 0.01302 & -0.004987\\
-0.004656  &  0.005729  &    0.9938 & -0.0003104 & -0.0007016 &  -0.003079 & 0.01889  & 0.009435\\
0.0032  &  -0.03714  &  0.009426  &    0.9808  & -0.003764  &  0.005483 & 0.01575  &  -0.2312\\
-0.004986  &   0.09533  &  -0.01069   &  0.05927   &   0.9833  &  0.009611 & -0.2954  &   0.9168\\
0.1371   &    -0.23   &   0.3707  &   -0.1227 &  -0.004499 & 0.805
&
-5.644  &   -2.599\\\hline 0.01302 & -0.007612  &  0.02119 &   -0.0378  & -0.06662  &   0.1193 & -0.3646 & -0.1998\\
0.01882 &  -0.02973  &  0.05876 &  -0.01479 &   0.01373   & 0.2482
& -0.778 &  -0.3508
\end{array}
\right].
\end{equation*}
The results of simulation of the closed-loop systems in conditions
of a windshear and noisy measurements are presented together with
the problem solution results in Table~\ref{table:TU-154} below and
illustrated in Figures~\ref{fig:TU154 output}--\ref{fig:TU154
disturb}. In the simulation we use a typical wind profile
described by the ring vortex downburst model~\cite{Ivan_1996}.

\begin{table}[thpb]
  \caption{TU-154 aircraft landing. Comparison of closed-loop systems}
  \begin{center}\label{table:TU-154}
\scriptsize
  \begin{tabular}{|cc|c|c|c|}
\hline & & \multicolumn{3}{|c|}{\underline{Controller in feedback loop}}\\
 & & $K_2$ & $K_a$ & $K_\infty$\\\hline
 \multicolumn{2}{c}{Solution results:}\\\hline
 $\min\gamma$ & & 0.516 & 5.4203 & 10.894\\
  $\|T_{ZW}\|_2$ &  & 0.516 & 1.1473 & 3.1448\\
$\sn T_{ZW}\sn_{0.7}$ &  & 7.8391 & 5.1768 & 5.5944\\
 $\|T_{ZW}\|_\infty$ & & 15.855 & 10.93 & 10.891\\
CPU time, & sec & 0.78001 & 5.928 & 1.7004\\
  \hline\multicolumn{2}{c}{Simulation results:}\\\hline
$\max{|\Delta V|},$ & m/sec & 11.3 & 3.559 & 4.329\\
$\max{|\Delta h|},$ & m & 54.79 & 46.87 & 39.79\\\hline
$\max{|\Delta\theta|},$ & deg & 14.86 & 16.04 & 31.6\\
$\max{|\Delta\omega_z|},$ & deg/sec & 4.884 & 5.043 & 10.56\\
$\max{|\Delta\vartheta|},$ & deg & 19.06 & 19 & 38.08\\
$\max{|\Delta T|},$ & kN & 7.263 & 22.58 & 42.48\\\hline
$\max{|\Delta\vartheta_{cy}|},$ & deg & 20.7 & 20.8 & 21.91\\
$\max{|\Delta\delta_t|},$ & deg & 8.224 & 29.25 & 29.23\\
\hline
  \end{tabular}
  \end{center}
\end{table}

From the solution results in Table~\ref{table:TU-154} we can
conclude that
\begin{itemize}
\item the respective minimum square root values of the objective
functions satisfy $\gamma_2 < \gamma_a < \gamma_\infty$; \item the
$a$-anisotropic norm of the closed-loop system with the
anisotropic $\gamma$-optimal controller satisfies $\sn
T_{ZW}\sn_{0.7}<\gamma_a;$ the controller is actually suboptimal.
\end{itemize}
Analysis of the simulation results presented in
Table~\ref{table:TU-154} and Figures~\ref{fig:TU154
output}--\ref{fig:TU154 disturb} shows that
\begin{itemize}
\item the anisotropic $\gamma$-optimal controller results in the
least maximal absolute deviation of the airspeed and admissible
maximal absolute deviation of the height; \item the worst maximal
absolute deviations of the controlled variables are demonstrated
by the $\H2$ optimal controller;
 \item the anisotropic
controller provides the maximal absolute deviation of the thrust
required for the manoeuvre \emph{almost two times less} than the
additional thrust required by the system with the $\Hinf$
controller; \item the same concerns the maximal absolute
deviations of the trajectory slope angle, pitch rate, and pitch;
\item the least maximal additional thrust is required by the
closed-loop system with the $\H2$ optimal controller; \item the
maximal values of the control signals of the anisotropic and
$\Hinf$ controllers are close, the control generated by the
anisotropic controller looks more realistic.
\end{itemize}

The anisotropic $\gamma$-optimal controller is obviously more
effective than the $\H2$ controller and less conservative than the
$\Hinf$ controller in this example of the disturbance attenuation
problem.

\begin{figure}[thpb]
      \centering
      \includegraphics[width=170mm]{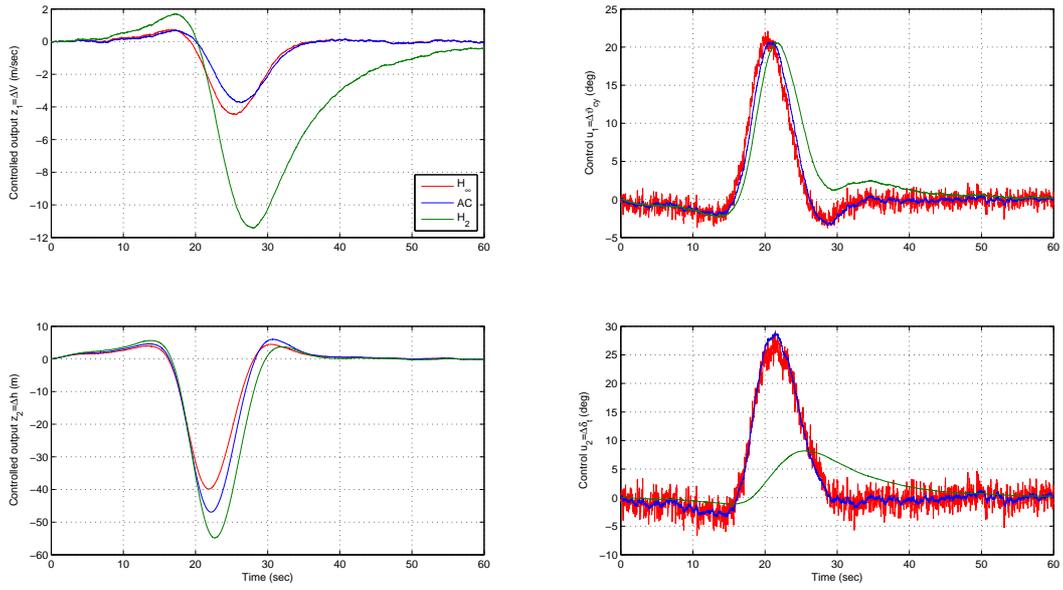}
      \caption{TU-154 aircraft landing. Airspeed $\Delta V$, height $\Delta h$ (left plots) and control signals $\Delta\vartheta_{cy}$, $\Delta\delta_t$ (right plots)}\label{fig:TU154 output}
\end{figure}

\begin{figure}[tphb]
      \centering
      \includegraphics[width=170mm]{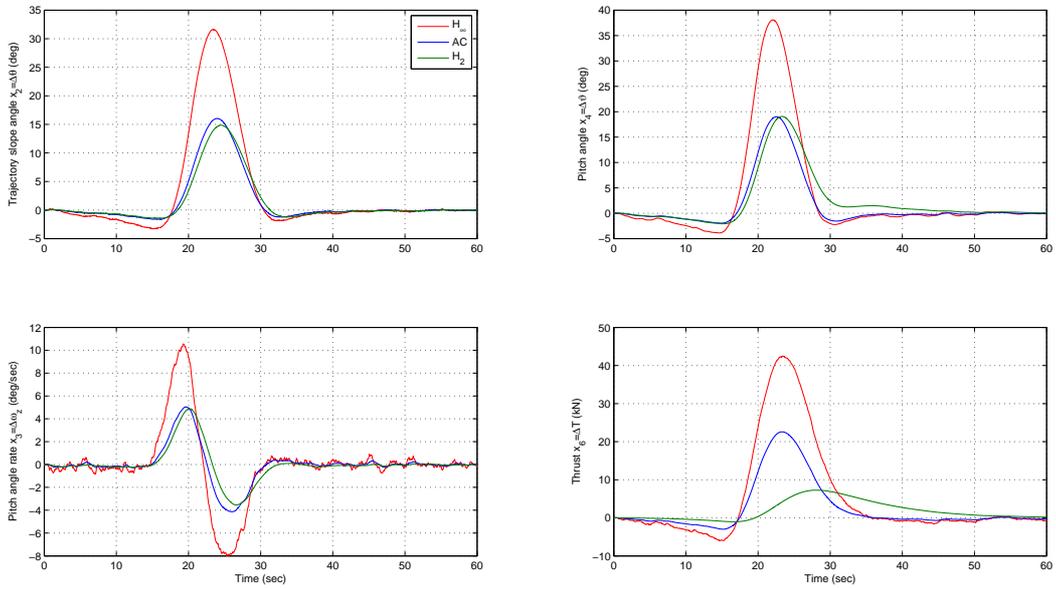}
      \caption{TU-154 aircraft landing. Trajectory slope angle $\Delta\theta$, pitch angle rate $\Delta\omega_z$ (left plots), pitch angle $\Delta\vartheta$, thrust $\Delta T$ (right plots)}\label{fig:TU154 state}
\end{figure}

\begin{figure}[tphb]
      \centering
      \includegraphics[width=170mm]{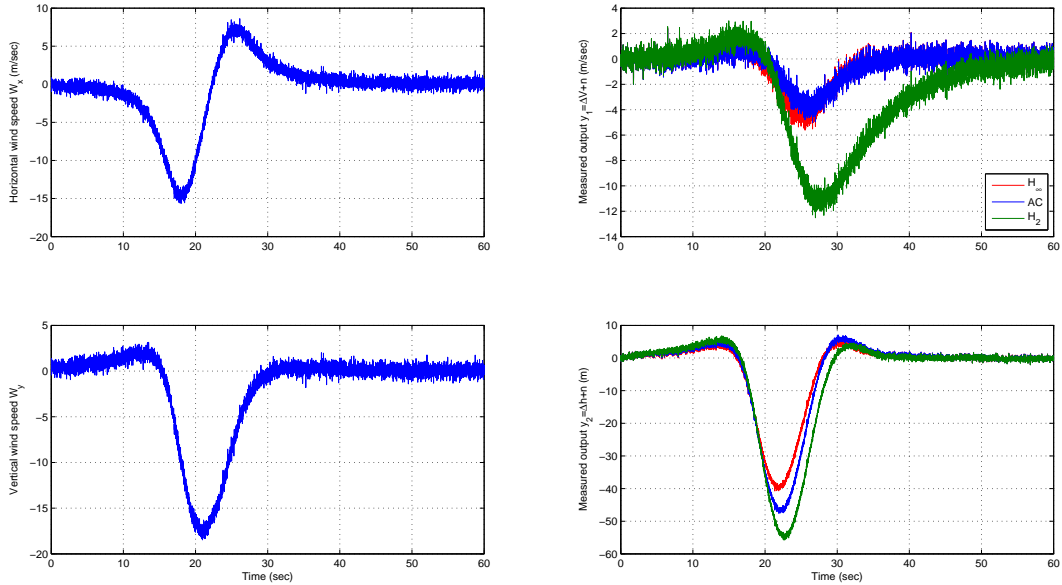}
      \caption{TU-154 aircraft landing. Wind profile (left plots) and noisy measurements (right plots)}\label{fig:TU154 disturb}
\end{figure}

\newpage

\subsubsection{$\bf{COMPl_eib}$ examples}

The anisotropic $\gamma$-optimal full-order controllers have been
computed for some models from the \emph{COMP$l_eib$}
collection~\cite{Leibfritz_2004,Leibfritz_Lipinski_2003} listed
below in Table~\ref{table:complib foc}. All of them were converted
from continuous- to discrete-time models with the sampling time
$\Delta t$. It is known from~\cite{Leibfritz_Lipinski_2003} that
almost all of these models (excepting ROC5) are SOF-stabilizable,
but here the respective problems are solved by the dynamic
full-order output-feedback controllers for the testing purpose
solely. In~\cite{TK_2011_prep} it is shown that satisfying the
conditions of SANBRL with $a\to0,+\infty$ ensures the $\H2$ and
$\Hinf$ norms not to exceed a given threshold value. Therefore the
$\H2$ and $\Hinf$ controllers for the respective problems have
also been derived as the limiting cases of the anisotropic
controller from a solution to the convex optimization
problem~(\ref{eq:minimum gamma2 fo of}) as defined in
Theorem~\ref{theorem:full-order output-feedback problem solution}
but with the respective input mean anisotropy levels $a=0$ and
$a=+\infty$.

\begin{table}[thpb]
  \caption{Examples from the \emph{COMP$l_eib$}
collection~\cite{Leibfritz_2004,Leibfritz_Lipinski_2003}.
Full-order design}
  \label{table:complib foc}
  \begin{center}\scriptsize
  \begin{tabular}{|c|c|c|c|c|c|c|c|c|c|}\hline
  Model & $(n_x,m_u,p_y)$ & $\Delta t$ & $\min\gamma_2$ & $a$ & $\min\gamma_a$ & $\min\gamma_\infty$ & \multicolumn{3}{|c|}{\underline{CPU time
  (sec)}}\\
  & & (sec) & & & & & $K_2$ & $K_a$ & $K_\infty$\\\hline
AC4 & $(4,1,2)$ & 0.0005 & 0.1782 & 0.015 & 0.20071 & 0.56227 & 0.92041 & 3.276 & 0.6864\\
AC7 & $(9,1,2)$ & 0.01 & 0.0042953 & 0.05 & 0.0094512 & 0.043755 & 2.7768 & 5.8032 & 2.0592\\
AC8 & $(9,1,5)$ & 0.01 & 0.049999 & 0.05 & 0.20454 & 1.5429 & 3.1668 & 5.7564 & 2.262\\
AC9 & $(10,4,5)$ & 0.01 & 0.04454 & 0.9 & 0.43057 & 1.0007 & 5.4912 & 13.9 & 5.2572\\
AC12 & $(4,3,4)$ & 0.01 & 0.0012071 & 0.01 & 0.0037555 & 0.31439 & 3.6504 & 3.4632 & 1.4508\\
HE3 & $(8,4,6)$ & 0.01 & 0.081028 & 0.015 & 0.18837 & 0.802 & 5.5068 & 2.964 & 3.1512\\
HE5 & $(8,4,2)$ & 0.01 & 0.11888 & 0.2 & 0.67939 & 1.5066 & 5.2104 & 2.3244 & 1.7784\\
HE6 & $(20,4,6)$ & 0.01 & 0.65791 & 0.05 & 0.78951 & 2.3755 & 22.745 & 25.007 & 21.202\\
HE7 & $(20,4,6)$ & 0.01 & 0.55239 & 0.05 & 0.68603 & 2.4341 & 24.087 & 24.679 & 21.481\\
JE1 & $(30,3,5)$ & 0.01 & 0.76355 & 0.1 & 1.1173 & --- & 287.03 & 345.65 & --- \\
JE3 & $(24,3,6)$ & 0.01 & 1.107 & 0.07 & 1.2814 & 2.4149 & 86.035 & 96.05 & 75.192 \\
EB1 & $(10,1,1)$ & 0.001 & 0.044894 & 3 & 3.0259 & 3.1041 & 5.9436 & 3.3228 & 2.6988\\
EB2 & $(10,1,1)$ & 0.001 & 0.027729 & 3 & 1.7246 & 1.7677 & 5.772 & 4.5552 & 3.9312\\
EB3 & $(10,1,1)$ & 0.001 & 0.029817 & 0.3 & 0.92218 & 1.7974 & 5.9748 & 3.198 & 3.1044\\
EB4 & $(20,1,1)$ & 0.001 & 0.030079 & 0.3 & 0.9219 & 1.7863 & 20.592 & 20.202 & 16.614\\
EB5 & $(40,1,1)$ & 0.001 & 0.029731 & 0.3 & 0.92087 & 1.7906 & 1042.6 & 1575.3 & 1258.5\\
ROC5 & $(7,3,5)$ & 0.001 & 0.0029492  & 0.7 & 0.0013201 & 0.0016873 & 1.3416 & 5.0856 & 1.6848\\
TF1 & $(7,2,4)$ & 0.1 & 0.043013 & 0.25 & 0.18306 & 0.24883 & 1.3884 & 3.7284 & 1.7004\\
TF3 & $(7,2,3)$ & 0.1 & 0.043081 & 0.25 & 0.18288 & 0.24799 &
1.4664 & 4.2588 & 1.8564\\\hline
  \end{tabular}
  \end{center}
\end{table}

Below we present the solution and simulation results for the
autopilot control problem for an air-to-air missile (AC4)
initially presented in~\cite{FAN_2001}, where this problem is
considered in robust setting and requires that the autopilot
generates the tail deflection $\delta$ to produce an angle of
attack $\alpha$ corresponding to a manoeuvre defined by the
guidance law~\cite{FAN_2001}. More precisely, the control aims at
tracking step input commands $\alpha_c$ with a steady state
accuracy of 1\%, achieving a rise time less than 0.2\,sec and
limiting overshoot to 2\% over a range of angles of attack
$\pm20$\,deg and variations in Mach number 2.5 to
3.5~\cite{FAN_2001}. The model AC4 from the \emph{COMP$l_eib$}
collection~\cite{Leibfritz_2004} does not take into account the
variations in Mach number and therefore not include uncertain
parameters.

The state-space realization of the anisotropic $\gamma$-optimal
controller $K_a$ computed for the mean anisotropy level $a=0.015$
is presented below together with the realizations of the $\H2$ and
$\Hinf$ controllers $K_2$ and $K_\infty$:
\begin{equation*}\label{eq:AC4 H2cont}
K_2 = \left[\scriptsize
\begin{array}{cccc|cc}
0.9966  & 0.0008498  &  0.000411 & 8.217\cdot10^{-5} &
-0.005969 & -0.002409\\
0.01915   &   0.9961 &  -0.003387 &  0.0003514 & -1.051   &  0.1632\\
-0.06821  &  0.008199  &     0.994  &  0.004536 & 1.194  &  -0.5822\\
0.4085  &  -0.04165 &  -0.006731  &    0.9939 & 11.78 &
5.414\\\hline 0.003389  & -0.001572 & -0.0008248 & -0.0001071 &
-0.0005778 & -0.0004742
\end{array}
\right],
\end{equation*}
\begin{equation*}\label{eq:AC4 anicont}
K_a = \left[\scriptsize
\begin{array}{cccc|cc}
0.9971 &  0.0006249  & 0.0002408 & 9.455\cdot10^{-5} & -0.02163 & 0.006044\\
0.006976   &   0.9975 &  -0.003433  &  0.001949 & -2.477 & -0.03384\\
-0.0344  &  0.005264   &    0.993  &  0.009394 & 3.793  & -0.1013\\
0.2344  &  -0.04362 &  -0.001715   &   0.9836 & 27.38 &
6.235\\\hline 0.002583 &  -0.001028 & -0.0004369 &  -0.000115 &
-0.002215 & -0.0006102
\end{array}
\right],
\end{equation*}
\begin{equation*}\label{eq:AC4 Hinfcont}
K_\infty = \left[\scriptsize
\begin{array}{cccc|cc}
0.9702 &  0.001709 & 0.0001139 & 0.0001427 & 0.07424 & -0.6099\\
0.3737  &   0.9805 & -0.000889  &  0.01253 & 6.54  &  9.539\\
-0.5466 &  0.009232  &   0.9905  &  0.02673 & 14.21  &  -11.1\\
0.2451  &  -0.1887 &   0.06776 &  -0.07204 & 210.2 & 29.54\\\hline
-7.581\cdot10^{-5}  & -0.0006162 &  -0.0002021  & -0.0002354 &
-0.1208 & -0.05091
\end{array}
\right].
\end{equation*}

The results of simulation of the closed-loop systems with the
$\H2$, anisotropic and $\Hinf$ controllers $K_2$, $K_a$ and
$K_\infty$ in conditions of
 noisy measurements are
illustrated in Figures~\ref{fig:missile alpha
delta}--\ref{fig:missile Bode step}. In simulation we generated
the reference commands $\alpha_c$ as steps with random amplitudes
and equal fixed durations. The step responses in
Figure~\ref{fig:missile Bode step} show that the closed-loop rise
time pointed out in these plots does not exceed the desired
0.2\,sec for all three controllers. At that, the diagrams of
Figure~\ref{fig:missile alpha delta} demonstrate acceptable
tracking performance and lesser amplitude of the control $\delta$
required in the closed-loop system with the anisotropic controller
in comparison with the $\Hinf$ controller.

\begin{figure}[thpb]
      \centering
      \includegraphics[width=170mm]{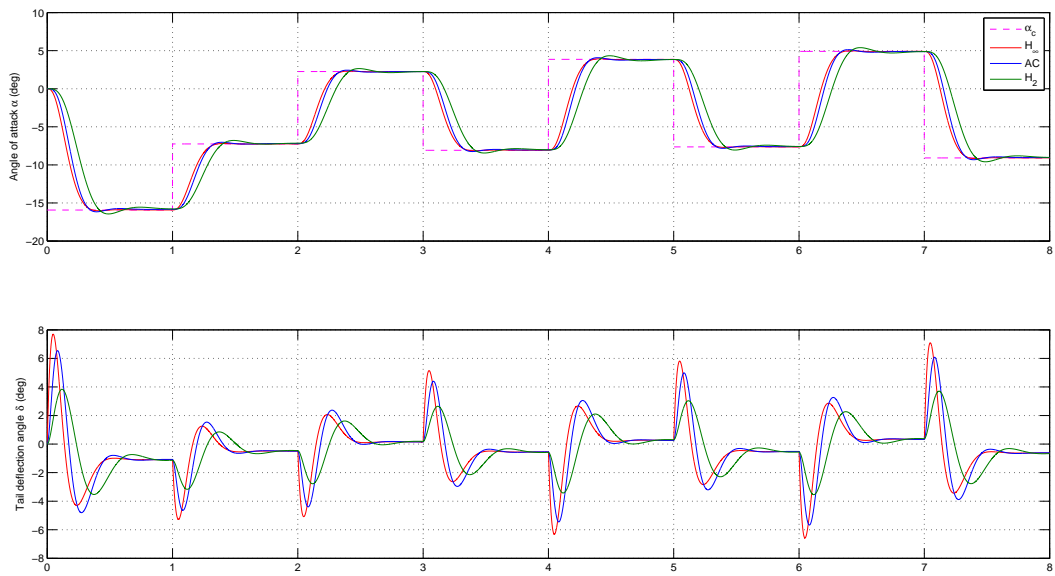}
      \caption{Model AC4 (air-to-air missile)~\cite{FAN_2001,Leibfritz_2004,Leibfritz_Lipinski_2003}. Angle of attack $\alpha$ (top plot) and tail deflection angle $\delta$ (bottom plot)}\label{fig:missile alpha delta}
\end{figure}

\begin{figure}[thpb]
      \centering
      \includegraphics[width=170mm]{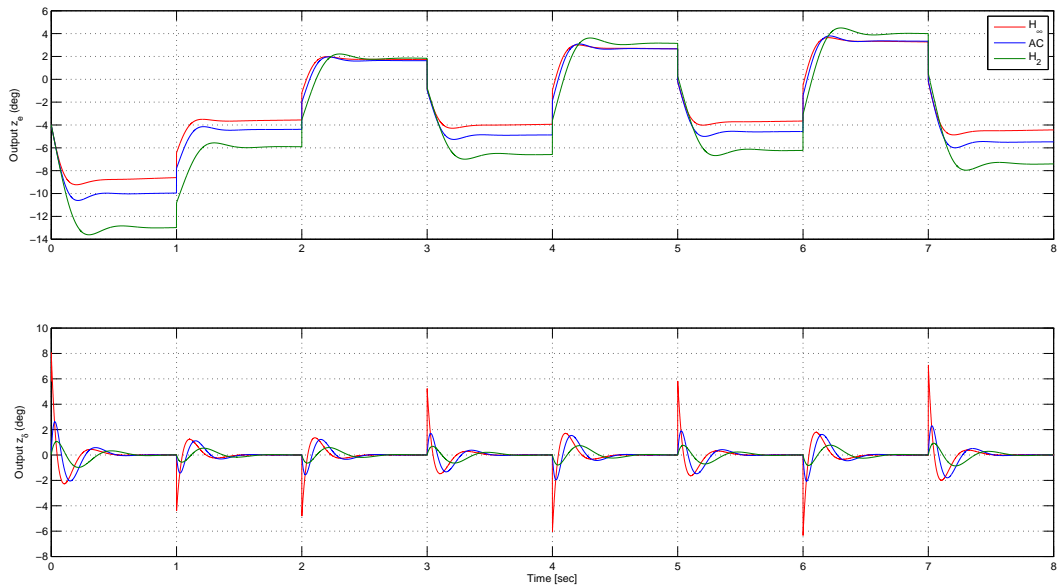}
      \caption{Model AC4 (air-to-air missile)~\cite{FAN_2001,Leibfritz_2004,Leibfritz_Lipinski_2003}. Controlled output $z$}\label{fig:missile z}
\end{figure}

\begin{figure}[thpb]
      \centering
      \includegraphics[width=170mm]{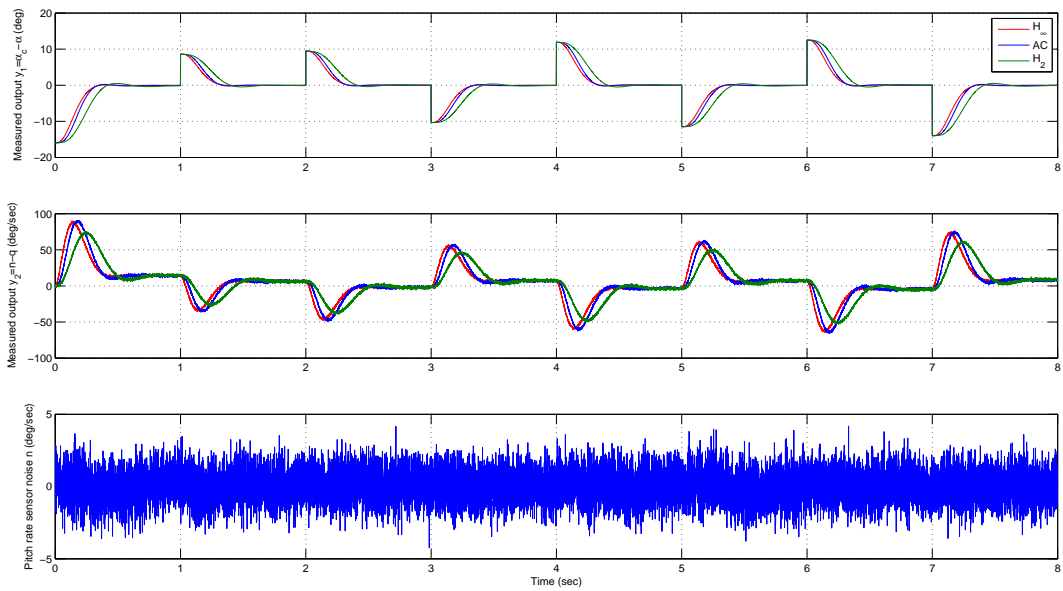}
      \caption{Model AC4 (air-to-air missile)~\cite{FAN_2001,Leibfritz_2004,Leibfritz_Lipinski_2003}. Measurement $y$ (top plots) and pitch rate sensor noise $n$ (bottom plot)}\label{fig:missile y n}
\end{figure}

\begin{figure}[thpb]
      \centering
      \includegraphics[width=170mm]{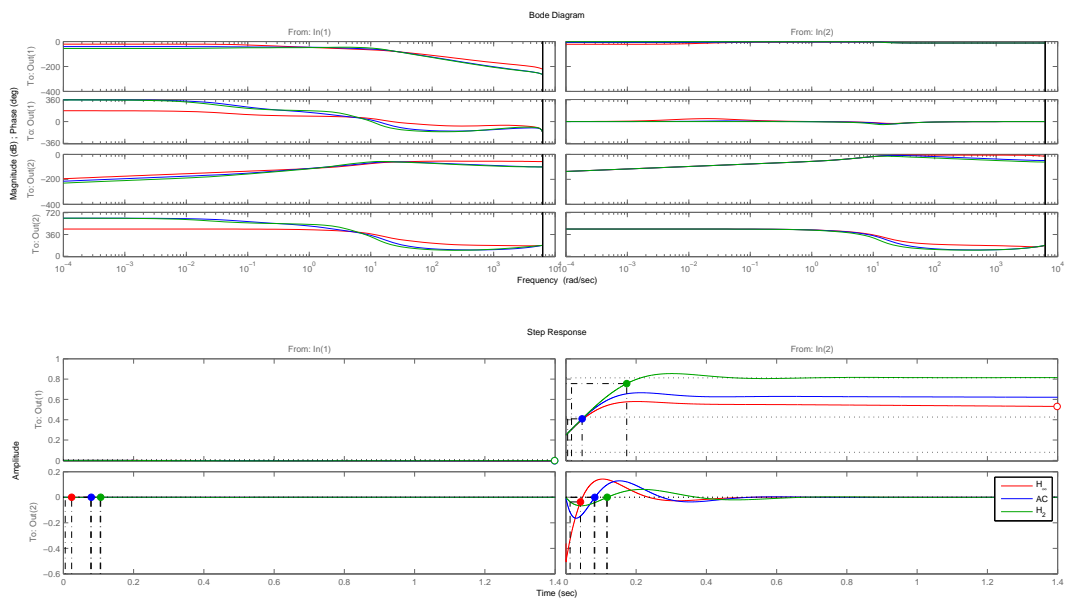}
      \caption{Model AC4 (air-to-air missile)~\cite{FAN_2001,Leibfritz_2004,Leibfritz_Lipinski_2003}. Bode diagram (top plots) and step response (bottom plots)}\label{fig:missile Bode step}
\end{figure}

\newpage

\subsection{Static output-feedback design}

The anisotropic $\gamma$-optimal static output-feedback
controllers have been computed for a number of singular filtering
problems from the \emph{COMP$l_eib$}
collection~\cite{Leibfritz_2004,Leibfritz_Lipinski_2003} listed
below in Table~\ref{table:complib sof}. As above, all of them were
converted from continuous- to discrete-time models with the
sampling time $\Delta t$. The $\H2$ and $\Hinf$ controllers for
the respective problems have been derived as the limiting cases of
the anisotropic controller from a solution to the convex
optimization problem~(\ref{eq:minimum gamma2 sof sf}) as defined
in Theorem~\ref{theorem:singular filtering sof convex} and
Corollary~\ref{corollary:sof problem convex opt sf} but with the
respective input mean anisotropy levels $a=0$ and $a=+\infty$.

\begin{table}[thpb]
  \caption{Examples from the \emph{COMP$l_eib$}
collection~\cite{Leibfritz_2004,Leibfritz_Lipinski_2003}. Static
output-feedback design}
  \label{table:complib sof}
  \begin{center}\scriptsize
  \begin{tabular}{|c|c|c|c|c|c|c|c|c|c|}\hline
  Model & $(n_x,m_u,p_y)$ & $\Delta t$ & $\min\gamma_2$ & $a$ & $\min\gamma_a$ & $\min\gamma_\infty$ & \multicolumn{3}{|c|}{\underline{CPU time
  (sec)}}\\
  & & (sec) & & & & & $K_2$ & $K_a$ & $K_\infty$\\\hline
AC1 & $(5,3,3)$ & 0.01 & 0.00045695 & 0.9 & 0.0034448 & 0.0036873 & 0.81121 & 3.042 & 0.546\\
AC2 & $(5,3,3)$ & 0.01 & 0.021254 & 0.9 & 1.3559 & 1.6199 & 0.99841 & 3.432 & 0.6552\\
AC15 & $(4,2,3)$ & 0.0001 & 0.037899 & 0.8 & 0.67708 & 0.79834 & 3.588 & 1.0764 & 0.702\\
HE1 & $(4,2,1)$ & 0.0001 & 0.00075643 & 0.15 & 0.0063848 & 0.0099472 & 2.8548 & 0.546 & 0.5148\\
HE4 & $(8,4,6)$ & 0.01 & 2.8727 & 0.05 & 8.0104 & 21.823 & 5.8812 & 3.6816 & 1.7784\\
NN15 & $(3,2,2)$ & 0.001 & 0.015202 & 0.3 & 0.25514 & 0.3441 & 3.0888 & 0.7488 & 0.81121\\
NN16 & $(8,4,4)$ & 0.001 & 0.0098319 & 0.5 & 0.20576 & 0.41639 & 4.0872 & 1.716 & 0.93601\\
BDT1 & $(11,3,3)$ & 1 & 0.010557 & 0.007 & 0.042299 & 0.32302 & 4.4928 & 1.3884 & 1.2168\\
PSM & $(7,2,3)$ & 0.001 & 0.035481 & 0.01 & 0.10554 & 0.92672 & 1.7784 & 4.134 & 0.7644\\
UWV & $(8,2,2)$ & 0.001 & 0.016479 & 0.03 & 0.011414 & 0.024207 &
3.1824 & 1.0452 & 0.90481\\\hline
  \end{tabular}
  \end{center}
\end{table}

For the purely illustrative purpose, below we present the solution
and simulation results for the aircraft control problem (AC1)
initially considered in~\cite{Hung_MacFarlane_1982}. The model AC1
from the \emph{COMP$l_eib$} collection~\cite{Leibfritz_2004} is
recast into a disturbance attenuation singular filtering problem
with noiseless measurements. The anisotropic $\gamma$-optimal
static gain $K_a$ computed for the mean anisotropy level $a=0.9$
is presented below together with the $\H2$ and $\Hinf$ gains $K_2$
and $K_\infty:$
\begin{equation*}\label{eq:AC1 H2cont}
K_2 = \left[\scriptsize
\begin{array}{ccc}
7.278\cdot10^{-5}   &  -0.9994 &  -0.000203\\
-0.0002887  & -0.002706  &   -0.9966\\
-0.9871   &   -14.34   &    49.35
\end{array}
\right],
\end{equation*}
\begin{equation*}\label{eq:AC1 anicont}
K_a = \left[\scriptsize
\begin{array}{ccc}
2.935\cdot10^{-6}    &       -1 & -1.241\cdot10^{-5}\\
-1.025\cdot10^{-5}  & -0.0001273  &    -0.9998\\
-0.5795   &    -12.46    &    54.56
\end{array}
\right],
\end{equation*}
\begin{equation*}\label{eq:AC1 Hinfcont}
K_\infty = \left[\scriptsize
\begin{array}{ccc}
-4.024\cdot10^{-7}     &      -1 &  1.904\cdot10^{-6}\\
-5.207\cdot10^{-6} & -6.032\cdot10^{-5}   &   -0.9999\\
-0.6788   &    -13.12    &    53.39
\end{array}
\right].
\end{equation*}

The results of simulation of the closed-loop systems in conditions
of a windshear are presented together with the problem solution
results in Table~\ref{table:AC1} below and illustrated in
Figures~\ref{fig:AC1 out contr}--\ref{fig:AC1 responses}. In the
simulation we use the same wind profile as in the example of
TU-154 aircraft flight control in Section~\ref{subsubsect:TU-154}.

\begin{table}[thpb]
  \caption{Example AC1 (aircraft) from the \emph{COMP$l_eib$}
collection~\cite{Leibfritz_2004,Leibfritz_Lipinski_2003}.
Comparison of closed-loop systems}
  \begin{center}\label{table:AC1}
\scriptsize
  \begin{tabular}{|cc|c|c|c|}
\hline & & \multicolumn{3}{|c|}{\underline{Controller in feedback loop}}\\
 & & $K_2$ & $K_a$ & $K_\infty$\\\hline
 \multicolumn{2}{c}{Solution results:}\\\hline
 $\min\gamma$ & & 0.00045695 & 0.0034448 & 0.0036873\\
 $\|T_{ZW}\|_2$ &  & $2.6532\cdot10^{-5}$ & $1.2762\cdot10^{-6}$ & $6.3218\cdot10^{-7}$\\
$\sn T_{ZW}\sn_{0.9}$ &  & 0.00050863 & $2.3466\cdot10^{-5}$ & $1.1795\cdot10^{-5}$\\
  $\|T_{ZW}\|_\infty$ & & 0.00075676 & $3.5153\cdot10^{-5}$ & $1.7708\cdot10^{-5}$\\
CPU time, & sec & 0.81121 & 3.042 & 0.546\\
  \hline\multicolumn{2}{c}{Simulation results:}\\\hline
  $\max{|z_1|},$ & m/sec & $9.539\cdot10^{-5}$  & $4.941\cdot10^{-6}$ & $1.368\cdot10^{-6}$\\
  $\max{|z_2|},$ & deg & 0.0003134  & $1.437\cdot10^{-5}$   & $9.539\cdot10^{-5}$\\\hline
$\max{|x_1|},$ & m & 3.152 & 3.412  & 3.35\\
$\max{|x_2|},$ & m/sec & 0.1647 & 0.1108 & 0.124\\
$\max{|x_3|},$ & deg & 0.02948  & 0.0192  & 0.02172\\
$\max{|x_4|},$ & deg/sec & 0.008596  & 0.006704  & 0.006841\\
$\max{|x_5|},$ & m/sec & 0.406  & 0.278  & 0.3097 \\\hline
$\max{|u_1|},$ & $10^{-1}$\,deg & 0.1648  & 0.1108  & 0.124\\
$\max{|u_2|},$ & m/sec$^{2}$ & 0.0299 & 0.01922  & 0.02173\\
$\max{|u_3|},$ & deg & 0.2117  & 0.1355  & 0.154\\
\hline
  \end{tabular}
  \end{center}
\end{table}

The solution results presented in Table~\ref{table:AC1} shows that
\begin{itemize}
\item the respective minimum square root values of the objective
functions satisfy $\gamma_2 < \gamma_a < \gamma_\infty$; \item the
$a$-anisotropic norm of the closed-loop system with the
anisotropic $\gamma$-optimal static gain satisfies $\sn
T_{ZW}\sn_{0.9}<\gamma_a;$ the controller is actually suboptimal;
\item the $\H2$ and $\Hinf$ norms of the closed-loop systems with
the respective $\gamma$-optimal gains satisfy
$\|T_{ZW}\|_2<\gamma_2$, $\|T_{ZW}\|_\infty<\gamma_\infty;$ the
$\H2$ and $\Hinf$ controllers are actually suboptimal too.
\end{itemize}
The simulation results presented in Table~\ref{table:AC1} and
Figures~\ref{fig:AC1 out contr}--\ref{fig:AC1 responses} allow to
conclude that
\begin{itemize}
\item the anisotropic $\gamma$-optimal output-feedback static gain
leads to the least maximal absolute deviations of the forward
speed $x_2$, pitch angle $x_3$, pitch angle rate $x_4$, and
vertical speed $x_5$, at that the least maximal absolute deviation
of the height error $x_1$ is achieved with the $\H2$
$\gamma$-optimal static gain; \item the worst maximal absolute
values of the controlled output are demonstrated by $\H2$
$\gamma$-optimal static gain;
 \item the anisotropic
$\gamma$-optimal static gain leads to the least maximum absolute
amplitudes of the control signals.
\end{itemize}

\begin{figure}[thpb]
      \centering
      \includegraphics[width=170mm]{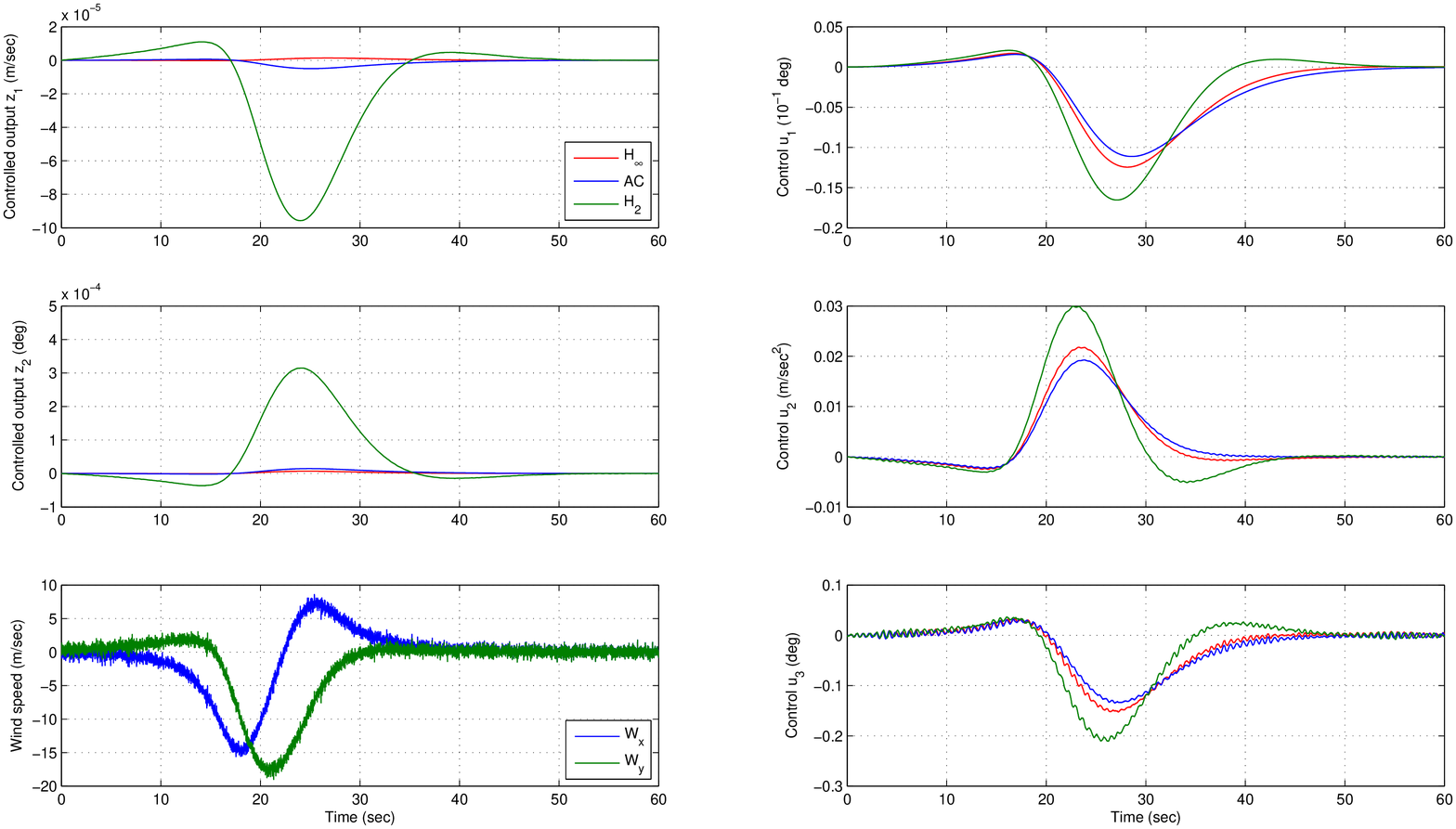}
      \caption{Model AC1 (aircraft)~\cite{Hung_MacFarlane_1982,Leibfritz_2004,Leibfritz_Lipinski_2003}. Controlled output and wind profile (left plots), control (right plots)}\label{fig:AC1 out contr}
\end{figure}

\begin{figure}[thpb]
      \centering
      \includegraphics[width=170mm]{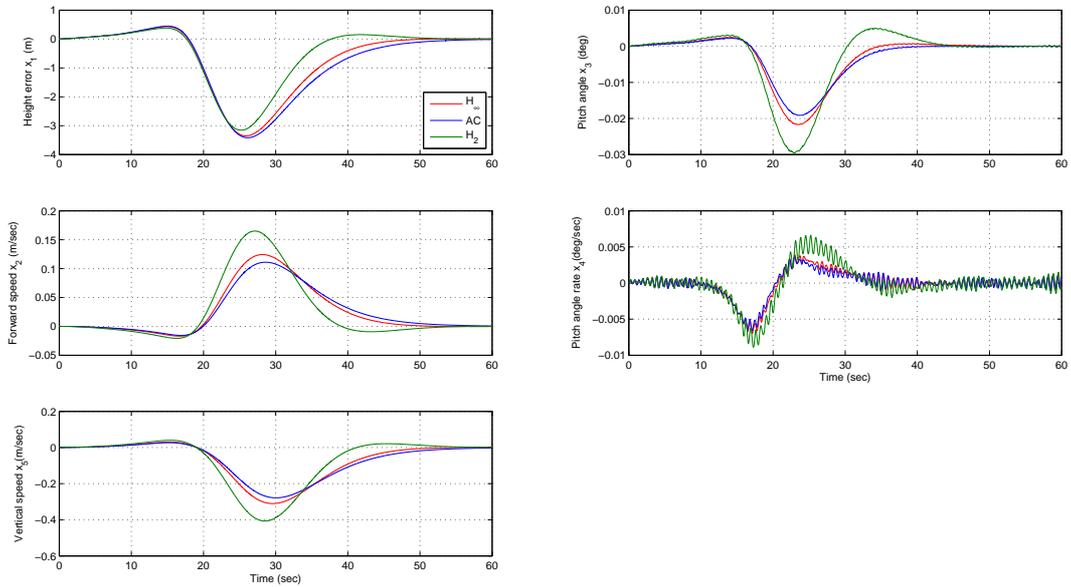}
      \caption{Model AC1 (aircraft)~\cite{Hung_MacFarlane_1982,Leibfritz_2004,Leibfritz_Lipinski_2003}. Height error $x_1$, forward speed $x_2$, vertical speed $x_5$ (left plots), pitch angle $x_3$, pitch angle rate $x_4$ (right plots)}\label{fig:AC1 state}
\end{figure}

\begin{figure}[thpb]
      \centering
      \includegraphics[width=170mm]{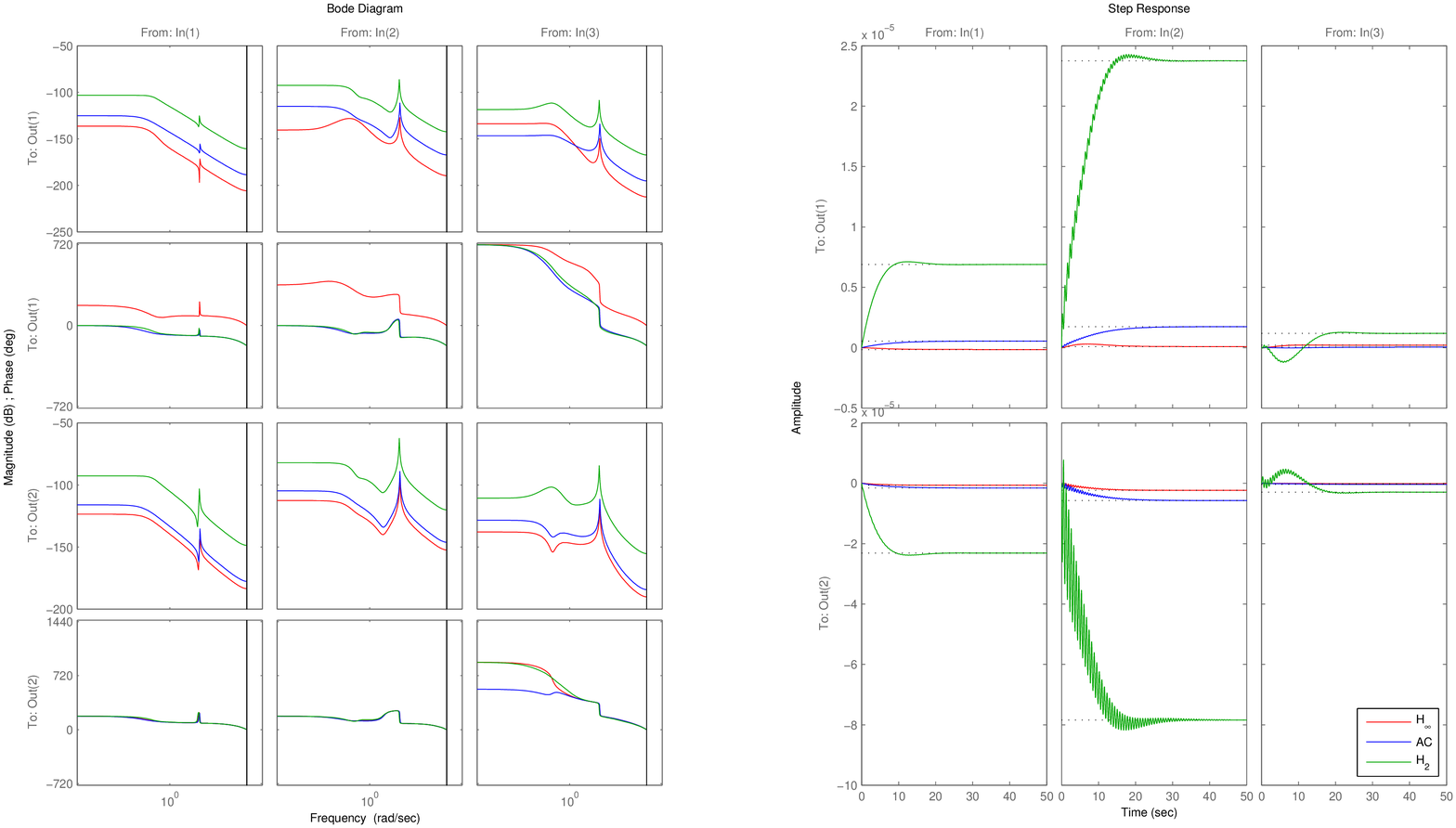}
      \caption{Model AC1 (aircraft)~\cite{Hung_MacFarlane_1982,Leibfritz_2004,Leibfritz_Lipinski_2003}. Bode diagram (left plot) and step response (right plot)}\label{fig:AC1 responses}
\end{figure}

\section{Conclusion}\label{sect:conclusion}

In this paper, we have proposed a solution to the anisotropic
suboptimal and $\gamma$-optimal controller synthesis problems by
convex optimization technique. The  anisotropic suboptimal
controller design is a natural extension of the optimal approach
developed in~\cite{VKS_1996_2}. Instead of minimizing the
anisotropic norm of the closed-loop system, the suboptimal
controller is only required to keep it below a given threshold
value. The general fixed-order synthesis procedure employs solving
an inequality on the determinant of a positive definite matrix and
two linear matrix inequalities in reciprocal matrices which make
the general optimization problem nonconvex. By applying the known
standard convexification procedures it have been shown that the
resulting optimization problem can be made convex for the
full-information state-feedback, output-feedback full-order
controllers, and static output-feedback controller for some
specific classes of plants defined by certain structural
properties. In the convex cases, the anisotropic $\gamma$-optimal
controllers are obtained by minimizing the squared norm threshold
value subject to convex constraints. In comparison with the
solution to the anisotropic optimal controller synthesis problem
derived in~\cite{VKS_1996_2} which results in a unique full-order
estimator-based controller defined by a complex system of
cross-coupled nonlinear matrix algebraic equations, the proposed
optimization-based approach is novel and does not require
developing specific homotopy-like computational algorithms.

\end{document}